\documentclass[a4paper]{paper}
\usepackage{amsmath,amsthm,amsfonts,mathrsfs}
\usepackage{amssymb}
\usepackage{tikz-cd}
\usepackage{xspace}
\usepackage{graphicx,booktabs}
\usepackage{paralist}
\usepackage{subfig}
\usepackage{hyperref}
\hypersetup{%
    pdfmenubar=true,       
    pdffitwindow=true,     
    pdfstartview={FitW},    
    pdftitle={MLAAS},
    colorlinks=true,       
    linkcolor=red,          
    citecolor=red,        
    filecolor=magenta,      
    urlcolor=cyan,           
}
\usepackage{acro}
\usepackage{cleveref}
\usepackage{algorithm,algorithmic}
\makeatletter
\def\BState{\State\hskip-\ALG@thistlm}
\makeatother
\Crefname{equation}{}{}
\usepackage{todonotes}
\usepackage{inputenc}
\usepackage{dsfont}

\newcommand{\bx}{\boldsymbol{x}}
\newcommand{\by}{\boldsymbol{y}}
\newcommand{\bdm}{\boldsymbol{m}}
\newcommand{\bds}{\boldsymbol{s}}
\newcommand{\bdlambda}{\boldsymbol{\lambda}}

\newcommand{\bdf}{\boldsymbol{f}}

\newcommand{\bdu}{\boldsymbol{u}}

\DeclareMathOperator*{\argmin}{arg\,min}

\newcommand{\tra}{\boldsymbol{\mathsf{T}}}

\newcommand{\Cdot}{\,\cdot\,}

\newcommand{\bdOneMask}{\boldsymbol{1}_{mask}}
\newcommand{\bdOne}{\boldsymbol{1}}
\newcommand{\bdvsig}{\boldsymbol{\varsigma}}
\newcommand{\proj}{\mathtt{Proj}}
\newcommand{\bz}{\boldsymbol{z}}

\theoremstyle{plain}
 \newtheorem{theorem}{Theorem}[section]
 \newtheorem{proposition}{Proposition}[section]
 \newtheorem{lemma}{Lemma}[section]

\theoremstyle{definition}
 \newtheorem{definition}{Definition}[section]

\theoremstyle{remark}
  \newtheorem{remark}{Remark}[section]
\newcommand{\bmu}{\boldsymbol{\mu}}
\newcommand{\Yspace}{\mathcal{Y}}
\newcommand{\Xspace}{\mathcal{X}}

\makeatletter

\newcommand{\Rmnum}[1]{\expandafter\@slowromancap\romannumeral #1@}
\makeatother

\begin{document}

\title{An efficient approach with theoretical guarantees to simultaneously reconstruct activity and attenuation sinogram for TOF-PET}
\author{Liyang Hu\thanks{State Key Laboratory of Mathematical Sciences, Academy of Mathematics and Systems Science, Chinese Academy of Sciences, Beijing 100190, China; University of Chinese Academy of Sciences, Beijing 100190, China.}
		\and
       Chong Chen\thanks{State Key Laboratory of Mathematical Sciences, Academy of Mathematics and Systems Science, Chinese Academy of Sciences, Beijing 100190, China.}
	}

\maketitle

\begin{abstract}
In positron emission tomography (PET), it is indispensable to perform attenuation correction in order to obtain the quantitatively accurate activity map  (tracer distribution) in the body. Generally, this is carried out based on the estimated attenuation map obtained from computed tomography or magnetic resonance imaging. However, except for errors in the attenuation correction factors obtained, the additional scan not only brings in new radiation doses and/or increases the scanning time but also leads to severe misalignment induced by various motions during and between the two sequential scans. To address these issues, based on maximum likelihood estimation, we propose a new mathematical model for simultaneously reconstructing the activity and attenuation sinogram from the time-of-flight (TOF)-PET emission data only. Particularly, we make full use of the exclusively exponential form for the attenuation correction factors, and consider the constraint of a total amount of the activity in some mask region in the proposed model. Furthermore, we prove its well-posedness, including the existence, uniqueness and stability of the solution. We propose an alternating update algorithm to solve the model, and also analyze its convergence. Finally, numerical experiments with various TOF-PET emission data demonstrate that the proposed method is of numerical convergence and robust to noise, and outperforms some state-of-the-art methods in terms of accuracy and efficiency, and has the capability of autonomous attenuation correction. 

\keywords{TOF-PET, autonomous attenuation correction, simultaneous reconstruction of activity and attenuation sinogram, maximum likelihood estimation, alternating update algorithm, theoretical guarantees}
\end{abstract}

\section{Introduction}\label{sec:Introduction}

Positron emission tomography (PET) is an important noninvasive functional imaging modality in nuclear medicine by injecting a dose of compounds labelled with a  radioactive tracer isotope (e.g., fluorine-18) that decays by positron emission into the patient, which has been widely used for  neurological, oncological and cardiovascular applications \cite{ollinger1997}. Once the isotope suffers from radioactive decay, the emitted positron travels a short distance of about 1-2 mm, following an annihilation with an electron, and simultaneously sending out a pair of high-energy annihilation photons propagating along approximately opposite directions, which subsequently are detected in a way of coincidence counting by the PET scanner surrounding the imaged object \cite{Vaquero2015Positron}. The acquisition of all these pairs consists of the PET measured data that is used to quantitatively and accurately reconstruct the distribution of the tracer activity. However, due to the absorption of photons by human tissues, the measured emission data is attenuated, in which case the attenuation correction is mandatory in order to avoid severely deteriorating the quality of reconstructed images. In general, an additional and separate transmission scan is conducted to provide the required attenuation correction factors \cite{burger2002,zaidi2006}. Recently, with the development of the nuclear medicine, the integration of PET imaging with other structural imaging modalities such as X-ray computed tomography (CT) \cite{kinahan2003,xia2011} or magnetic resonance imaging (MRI) \cite{bezrukov2013,burgos2014} has gained more and more interests and provides a novel strategy for evaluating the desirable attenuation map. But, in this way, there are still some limitations and challenges. For instance, in current PET-CT imaging systems, the attenuation correction is normally executed by making use of the well-matched CT images that are regulated from the relatively low energy range of X-ray to the photon energy of 511 keV, in which situation the accuracy of registration of PET and CT is typically affected by motions \cite{osman2003}. It is more complicated to synthesize the CT-equivalent attenuation maps from PET-MR imaging systems by means of the multi-atlas information. The process traditionally involves the use of locally alignment between the MRI-based patient’s morphology to a known database of MRI/CT pairs by a local image
similarity measure. But as known, there is a dilemma that the difference between the gray values of bone and that of air in MRI is extremely small instead the difference in PET is substantially significant.

To overcome these drawbacks mentioned previously, some authors have suggested to perform attenuation correction only using the emission data. An early attempt to settle this problem can date back to the work of Censor et al. \cite{censor1979}, simply treating it as an underdetermined nonlinear system and solving it by an iterative approach. Another attempt was based on the Helgason--Ludwig data consistency conditions \cite{helgason1965,ludwig1966}. Natterer et al. \cite{natterer1992} provided some theoretical results based on the strong assumptions on attenuation map to answer in which case the attenuation map can be determined from emission scans. The estimation of attenuation map was converted to compute an affine transformation of a known attenuation map \cite{natterer1993, welch1997}. Indeed, these results are of limited practical success. A different idea to determine attenuation map from emission data was proposed by Bronnikov \cite{bronnikov1995}, assuming that the activity map is approximately uniform and attenuation coefficients are small. Subsequently, the author again developed a two-stage approach based on the discrete consistency in \cite{bronnikov2000}, firstly estimating the attenuation image from an approximately discrete linear model and following the attenuation compensation.

Alternatively, different from the described methods above, statistical based models take sufficiently the Poisson nature of the collected data into account, tending to retrieve better reconstruction performance. Typically, the method of maximum likelihood (ML) is often used to solve such kind of models. More concretely, the authors in \cite{krol1995} devised an EM algorithm to maximize the likelihood for simultaneous reconstruction of activity and attenuation in single photon emission computed tomography. Nuyts et al. extended this approach to PET imaging exploiting a maximum a posterior (MAP) framework by incorporating some additional prior information about the attenuation coefficients \cite{nuyts1999}. Another related work can be found in \cite{clinthorne1991}, where an algorithm was derived for joint ML estimation of attenuation and activity images by the emission and transmission measurements. In addition, in \cite{depierro2006}, De Pierro et al. developed to alternatively maximize a surrogate function of the log-likelihood function for simultaneous activity and attenuation recovery from emission data only. Nevertheless, nearly all the above studies reported the so-called cross-talk phenomenon, where the local errors in the attenuation map are closely intertwined with the local errors in the activity image, especially occurring in the edge of activity image. For the ML approach, the existence of local solutions accounts for the disagreeable phenomenon.

Fortunately, with the advent of the time-of-flight (TOF) technique in PET \cite{lewellen1998}, that phenomenon can be effectively alleviated to large extent. The differential timing of the detection of the two photons can be used to localize the annihilation along the line of response (LOR) in TOF-PET imaging systems.  On the one hand, the availability of TOF information enables to accelerate existing iterative reconstruction algorithms \cite{wang2006}. On the other hand, it can help to improve the contrast-to-noise ratio of reconstructed images \cite{karp2008,moses2003}. Moreover, some studies demonstrated that the use of TOF can also strikingly relieve the artifacts induced by attenuation correction errors \cite{conti2010}. These benefits taken by TOF technique encourage researchers to devote more attention to the simultaneous estimation issue in TOF-PET. In \cite{defrise2012}, Defrise et al. provided a result that the attenuation sinogram can be determined by the emission data alone up to a constant based on the consistency conditions for the continuous TOF-PET model. At the same time, Rezaei et al. \cite{rezaei2012} explained why the use of TOF information can eliminate the cross-talk problem. Furthermore, they developed a study of the joint recovery of activity and attenuation correction factors for the discrete TOF-PET model using ML \cite{rezaei2014}, and Li et al. proposed a joint estimation of activity and attenuation sinogram using the consistency condition of the TOF-PET data \cite{li2017joint}.

\paragraph{Main contributions.} In this article, we study the simultaneous reconstruction of the activity and attenuation sinogram using only the TOF-PET emission data by optimizing a discrete negative log-likelihood estimation with some appropriate constraints, and propose an efficient and robust algorithm accordingly, which is actually a Maximum Likelihood reconstruction for joint Activity and Attenuation Sinogram, referred to as MLAAS. Additionally, the thorough and detailed theoretical analyses of the proposed model and algorithm are also provided. To our best knowledge, it is the first literature to prove the existence, uniqueness and stability for such kind of simultaneous reconstruction problems in discrete case. In contrast to \cite{rezaei2012}, there is no need for the proposed approach to reconstruct the attenuation map. Compared with \cite{rezaei2014}, the proposed model can exploit the intrinsically exponential form, a prior knowledge about the attenuation correction factors, which can justify a faster and more stable convergent performance. Furthermore, the current work introduces a weaker constraint for the total amount of the activity in some mask region, rather than the whole field of view in \cite{ren2024}, to overcome the problem of solutions differing up to a constant.

\paragraph{Outline.} The rest of the paper is organized as follows. In \cref{sec:MM}, we propose the considered mathematical model and prove its well-posedness, including solution existence, uniqueness and stability. Next, in \cref{sec:method}, the computational method for solving the proposed model is presented, and its convergence analysis is also studied. The relationship to some existing models is given in \cref{sec:comparison}. We conduct several numerical experiments to assess the effectiveness of the proposed method and to make numerical comparisons with those models in \cref{sec:experiments}. In terms of these results, we give a necessary discussion in \cref{discussion}. Finally, we conclude the work in \cref{conclusion}.

\section{The proposed mathematical model and its well-posedness}
\label{sec:MM}

First of all, we introduce some notations which are requisite for description. Let $\mathbb{R}^D$ be a $D$-dimensional Euclidean space with vector as $\bx = [x_1, \ldots, x_D]^{\tra}$, which equips with an inner product $\langle \bx, \by\rangle = \bx^{\tra}\by$ and the norm $\|\cdot\| = \sqrt{\langle\Cdot, \Cdot\rangle}$. Define $\mathbb{R}_{+}^D := \{\bx \in \mathbb{R}^D | x_d \ge 0, d = 1, \dots, D\}$, $\mathbb{N}$ as the set of all positive integers, and $\boldsymbol{1}$ as the vector composed of unit entries as the same dimension as the object of operation. Write $\bx \ge (>)~c$ or $\le (<)~c$ as each element of itself is not less (greater) than $c$ or not greater (less) than $c$. Denote Kullback--Leibler (KL) divergence by 
\begin{equation}\label{eq:KL_div}
\mathbf{KL}(a, b) := b - a + a\ln (a/b) \ge 0 \quad \text{for}~a, b >0. 
\end{equation}

\subsection{The proposed mathematical model}
\label{sec:model}

Generally, in TOF-PET, for LOR $i$ and TOF bin $t$, denoted by $(i, t)$, the collected count $M_{it}$ can be modeled as a random variable with expectation \begin{equation}\label{mean_data}
    \bar{m}_{it} = \exp{(-s_i)}\sum_{j=1}^{J}a_{ijt}\lambda_{j}, \quad i=1,\ldots,I, ~t=1,\ldots,T, 
\end{equation}
where $s_i$ denotes the sinogram element of the attenuation image of interest on LOR $i$, the whole $\exp{(-s_i)}$ is served as the associated attenuation correction factor, $\lambda_{j}$ is the activity value at pixel/voxel $j$, and $a_{ijt}$ represents the element of TOF sensitivity matrix at measurement $(i, t)$ for activity $j$ in the absence of attenuation. Note that the expected additive contribution due to scatter and/or randoms is not considered in this work. In order to better understand the formula of \cref{mean_data}, we recommend referring to its continuous form in \cite{defrise2012}.  


Here, the random variables $\{M_{it}\}_{i, t}$ are assumed to be independent of each other, and are Poisson distributed respectively as 
\begin{equation}
    \label{Poisson}
    M_{it}\sim \textrm{Poisson}(\bar{m}_{it}). 
\end{equation}
Let $m_{it}$ be the measured data or a realization of $M_{it}$, namely, the number of events detected for $(i, t)$. Assuming that $m_{it} > 0$ and $\bar{m}_{it} > 0$ for any $(i, t)$, which means that each LOR contains activity. Denote $\bdm = \{m_{it}\}_{i,t}$, and $\bar{\bdm} = \{\bar{m}_{it}\}_{i,t}$. Using \cref{mean_data}, the negative log-likelihood function can be given as 
\begin{align} 
 \nonumber   L(\bdlambda, \bds; \bdm) &= \sum_{i,t}\{\bar{m}_{it} - m_{it}\ln{\bar{m}_{it}}\} \\ 
\label{neg_log_likelihood}   &= \sum_{i,t} \left\{\exp{(-s_i)}\sum_{j=1}^{J}a_{ijt}\lambda_{j} - m_{it}\ln\left(\exp{(-s_i)}\sum_{j=1}^{J}a_{ijt}\lambda_{j}\right)\right\}, 
\end{align}
where $\bdlambda = [\lambda_1, \ldots, \lambda_J]^{\tra}$, and $\bds = [s_1, \ldots, s_I]^{\tra}$. 

In this article, the aim is to reconstruct activity $\boldsymbol{\lambda}$ and attenuation sinogram $\boldsymbol{s}$ simultaneously by minimizing \cref{neg_log_likelihood} under some appropriate constraints. 
More precisely, the proposed mathematical model, namely, maximum likelihood estimation for simultaneously reconstructing activity and attenuation sinogram can be formulated as
\begin{align}
\label{equ1} &\min_{\bdlambda , \bds} L(\bdlambda, \bds; \bdm) \\ 
\label{eq_constraint_1} &\hspace{1mm} \text{s.t.}  ~ 0 \le \bdlambda \le \Lambda, \\
\label{eq_constraint_2} &\hspace{7mm} \bds \ge 0, \\
\label{eq_constraint} &\hspace{7mm} \bdOneMask^{\tra}\bdlambda = N, 
\end{align}
where $\Lambda$ represents an upper bound of the activity, $\bdOneMask$ denotes a vector where the elements are $1$ at some specified positions in the mask region and $0$ elsewhere, and $N$ is the total amount of the activity in the mask region. 

\begin{remark}\label{rem:rk1}
Evidently, the function in \cref{neg_log_likelihood} is of scale invariance, i.e., 
\begin{equation}
\label{times relation}
    L(\bdlambda, \bds; \bdm) = L(\bdlambda \exp{(\gamma)}, \bds+\gamma\boldsymbol{1}; \bdm) \quad \forall \gamma\in\mathbb{R}.
\end{equation}
Hence, we introduce a weaker constraint of \cref{eq_constraint} than that in \cite{ren2024} to overcome the problem of solutions differing up to a constant. Although the objective function in \cref{equ1} has been considered in the work of \cite{li2017joint}, the proposed model and solving algorithm here are quite different, and the related theory is not established there either. 
\end{remark}

\subsection{Well-posedness}\label{set:wellposedness}

In this section, we prove the well-posedness of the proposed model \cref{equ1}-\cref{eq_constraint}, including the existence, uniqueness and stability of the solution. 

\subsubsection{Existence and uniqueness}\label{analysis}

Here we will demonstrate that the proposed model \cref{equ1}-\cref{eq_constraint} is exactly solvable, namely, there exists at least a minimizer that solves the considered problem. To this end, we first give the following result. 
\begin{lemma}
\label{Analemma1}
Given the measured data $\bdm > 0$, the function $L(\Cdot, \Cdot; \bdm)$ in \cref{neg_log_likelihood} is lower bounded in $\mathbb{R}_{+}^{J}\times\mathbb{R}_{+}^{I}$.
\end{lemma}
\begin{proof}
By simple derivation, and using \cref{eq:KL_div}, we have 
\begin{align*}
        L(\bdlambda, \bds; &\bdm) \\
        &= \sum_{i,t}\left\{\exp{(-s_i)}  \sum_{j=1}^{J}a_{ijt}\lambda_j - m_{it} + m_{it}\ln\left(\frac{m_{it}}{\exp{(-s_i)}\sum_{j=1}^{J}a_{ijt}\lambda_j}\right)\right\}\\
        &\hspace{4mm} +\sum_{i,t}\left\{m_{it} -m_{it}\ln m_{it}\right\}\\
        &=\sum_{i,t}\mathbf{KL}\left(m_{it}, \exp{(-s_i)}\sum_{j=1}^{J}a_{ijt}\lambda_{j}\right)+\sum\limits_{i,t}\left\{m_{it} -m_{it}\ln m_{it} \right\}\\
        &\ge \sum_{i,t}\left\{m_{it} -m_{it}\ln m_{it}\right\}. 
\end{align*}
Then, the desired result is proved. 
\end{proof} 

To proceed, we define $g_b(x) = x - b\ln x$ for $x>0$. Then the following proposition is valid. 
\begin{proposition}
\label{stability_lemma2}
Let $b>0$. For $\sigma >1$, there exist two positive numbers $K_1(\sigma, b) < K_2(\sigma, b)$ such that $g_b\bigl(K_1(\sigma,b)\bigr)=g_b\bigl(K_2(\sigma,b)\bigr)=\sigma$ and 
\[
[K_1(\sigma,b),K_2(\sigma,b)] = \{x > 0 | g_b(x) \le \sigma \}. 
\]
Furthermore, if $b_2>b_1>0$, then 
\begin{equation}
    K_1(\sigma,b_1) < K_1(\sigma,b_2)<1~\textrm{and}~ 1<K_2(\sigma,b_1) < K_2(\sigma,b_2).
\end{equation}
\end{proposition}

\begin{proof}
The first derivative of $g_b$ is calculated by   
\begin{equation*}
    g'_b(x) = 1- \frac{b}{x}.
\end{equation*}
Obviously, if $0 < x < b$, $g'_b(x) < 0$; if $x = b$, $g'_b(x) = 0$; if $x>b$, $g'_b(x) > 0$. This suggests that the function $g_b$ is strictly decreasing over $(0, b)$ but increasing over $(b, +\infty)$, and its unique minimum is taken at $x=b$. Furthermore, it is simple to verify that the minimum $g_b(b) \le 1$, and 
\begin{equation*}
    \lim_{x\to 0^{+}}g_b(x) = +\infty, \quad \lim_{x\to +\infty}g_b(x) = +\infty.
\end{equation*}
Consequently, for $\sigma > 1$, the equation $g_b(x) = \sigma$ is solvable, and admits two different solutions denoted by $K_1(\sigma,b) < K_2(\sigma,b)$. Combined with the distribution behavior of $g_b$, it turns out that $g_b(x)\le\sigma$ iff $x\in [K_1(\sigma,b),K_2(\sigma,b)]$. Note that for any $b$, $g_{b}(1) = 1$. Hence, for $b_2>b_1>0$, 
\begin{equation*}
    K_1(\sigma,b_1)<1<K_2(\sigma,b_1)~\textrm{and}~K_1(\sigma,b_2)<1<K_2(\sigma,b_2).
\end{equation*}
Since 
\begin{equation*}
    g_{b_2}(x)-g_{b_1}(x) = (b_1-b_2)\ln x,
\end{equation*}
we know that $g_{b_2}(x) > g_{b_1}(x)$ for $0<x<1$, and $g_{b_2}(x) < g_{b_1}(x)$ for $x>1$. Then 
\begin{equation*}
    \sigma = g_{b_1}(K_1(\sigma,b_1)) = g_{b_2}\bigl(K_1(\sigma,b_2)\bigr) > g_{b_1}\bigl(K_1(\sigma,b_2)\bigr),
\end{equation*}
which implies that $K_1(\sigma,b_1) < K_1(\sigma,b_2)$, similarly, 
\begin{equation*}
    g_{b_2}\bigl(K_2(\sigma,b_1)\bigr) < g_{b_1}\bigl(K_2(\sigma,b_1)\bigr) = g_{b_2}\bigl(K_2(\sigma,b_2)\bigr) = \sigma,
\end{equation*}
then $K_2(\sigma, b_1) < K_2(\sigma,b_2)$. Thus the proof is concluded. 
\end{proof}

For simplicity, let  
\[
\Xspace = \{\bdlambda \in \mathbb{R}^J | 0 \le \bdlambda \le \Lambda\}\times\mathbb{R}_{+}^{I}, 
\]
and $a_{max} = \max_{i,j,t}\{a_{ijt}\}$, $m_{min} = \min_{i,t}\{m_{it}\}$. 
\begin{lemma}\label{Analemma2}
Assume that $a_{max} > 0$, and the measured data $\bdm > 0$. The function $L(\Cdot, \Cdot; \bdm)$ of \cref{neg_log_likelihood} is coercive in $\Xspace$.
\end{lemma}
\begin{proof}
To prove the coercivity, we need to verify that if there exists a sequence $\{(\bdlambda^n, \bds^n)\}_{n\in\mathbb{N}} \in \Xspace$ such that $\{L(\bdlambda^{n}, \bds^{n}; \bdm)\}_{n\in\mathbb{N}}$ is bounded, then the sequence $\{(\bdlambda^{n},\bds^{n})\}_{n\in\mathbb{N}}$ is also bounded. 

Suppose that 
\begin{equation}\label{eq:L_upbound}
L(\bdlambda^{n}, \bds^{n}; \bdm) \le C \quad \forall n\in \mathbb{N}, 
\end{equation}
where $C$ is a positive constant. 

Since 
\begin{equation}\label{eq:upbound_lam}
    0 \le \bdlambda^n \le \Lambda \quad \forall n\in \mathbb{N}, 
\end{equation}
we have the estimate that 
\begin{equation}
\label{Anaproof2_fisrt}
    \|\bdlambda^{n}\|\le \sqrt{J}\Lambda. 
\end{equation}
In addition, since 
\begin{align*}
    L(\bdlambda^{n}, \bds^{n};& \bdm) \\
    &=  \sum_{i,t}\left\{\exp{(-s_{i}^{n})} \sum_{j=1}^{J}a_{ijt}\lambda_{j}^{n} - m_{it}\ln\left(\exp{(-s_{i}^{n})} \sum_{j=1}^{J}a_{ijt}\lambda_{j}^{n}\right)\right\}\le C,
\end{align*}
by the proof of \cref{Analemma1},  and recalling the definition of $\mathbf{KL}(\Cdot, \Cdot)$ in \cref{eq:KL_div}, it follows that 
\begin{equation*}
     \sum_{i,t}\mathbf{KL}\left(m_{it}, \exp{(-s_{i}^{n})}\sum\limits_{j=1}^{J}a_{ijt}\lambda_{j}^{n}\right) \le C - E, 
\end{equation*} 
where $E = \sum_{i,t}\{m_{it} -m_{it}\ln m_{it}\}$. Using the nonnegativity of $\mathbf{KL}(\Cdot, \Cdot)$, we obtain 
\begin{equation*}
     \mathbf{KL}\left(m_{it}, \exp{(-s_{i}^{n})}\sum\limits_{j=1}^{J}a_{ijt}\lambda_{j}^{n}\right) \le C - E \quad \forall (i, t). 
\end{equation*} 
From the proof of \cref{stability_lemma2}, we know that $g_b(b) \le 1$ for all $b>0$. Then, by the inequality above, we have    
\begin{equation}\label{Anaequ2}
    \exp{(-s_{i}^{n})}\sum_{j=1}^{J}a_{ijt}\lambda_{j}^{n} - m_{it}\ln\left(\exp{(-s_{i}^{n})}\sum_{j=1}^{J}a_{ijt}\lambda_{j}^{n}\right)\le \tilde{C} \quad\forall (i, t). 
\end{equation}
Here $\tilde{C} = C + IT +1$. By the proof of \cref{stability_lemma2}, from \cref{Anaequ2}, we obtain 
\begin{equation*}
   0 < K_1(\tilde{C}, m_{min}) \le \exp{(-s_{i}^{n})}\sum_{j=1}^{J}a_{ijt}\lambda_{j}^{n}\quad\forall (i, t). 
\end{equation*}
Using \cref{eq:upbound_lam}, and $a_{max} > 0$, we have 
\begin{equation*}
     0\le s_{i}^{n}\le \ln\left(\frac{J\Lambda a_{max}}{K_1(\tilde{C}, m_{min})}\right) \quad \forall i.
\end{equation*}
Furthermore,  
\begin{equation}
\label{Anaproof2_second}
    \|\bds^{n}\|\le \sqrt{I}\ln\left(\frac{J\Lambda a_{max}}{K_1(\tilde{C}, m_{min})}\right).
\end{equation}
Hence, the boundedness of the sequence $\{(\bdlambda^{n},\bds^{n})\}_{n\in\mathbb{N}}$ is immediately obtained by \cref{Anaproof2_fisrt} and \cref{Anaproof2_second}. 
\end{proof}

Subsequently, we can give the following results on the existence and uniqueness of the solution to the optimization problem 
\begin{equation}
\label{eq:min_L_X} \min_{(\bdlambda , \bds) \in \Xspace} L(\bdlambda, \bds; \bdm). 
\end{equation}

\begin{theorem}\label{existing_theorem}
Let the assumptions in \cref{Analemma2} hold. The solution to the problem in \cref{eq:min_L_X} exists.
\end{theorem}
\begin{proof}
By \cref{Analemma1}, we know that the function $L(\Cdot, \Cdot; \bdm)$ of \cref{neg_log_likelihood} is also lower bounded in $\Xspace$, which indicates the existence of a minimizing sequence of \cref{eq:min_L_X}. 

Let $\{(\bdlambda^n, \bds^n)\}_{n\in\mathbb{N}}$ be the minimizing sequence, namely, 
\begin{equation}\label{eq:mini_seq}
\lim_{n \to\infty} L(\bdlambda^n, \bds^n; \bdm) = \min_{(\bdlambda , \bds) \in \Xspace}L(\bdlambda, \bds; \bdm). 
\end{equation}
So the sequence $\{L(\bdlambda^n, \bds^n; \bdm)\}_{n\in\mathbb{N}}$ is bounded. Using \cref{Analemma2}, we have that $\{(\bdlambda^{n}, \bds^{n})\}_{n\in\mathbb{N}}$ is also bounded. Next, in the finite-dimensional Euclidean space, it follows that there exists a convergent subsequence $\{(\bdlambda^{n_k}, \bds^{n_k})\}_{k\in\mathbb{N}}$ such that 
\[
\lim_{k \to\infty}(\bdlambda^{n_k}, \bds^{n_k}) = (\bdlambda^{\ast},\bds^{\ast}) \in \Xspace. 
\]
By the continuity of the function $L(\Cdot, \Cdot; \bdm)$, and \cref{eq:mini_seq}, we obtain 
\begin{equation*}
    L(\bdlambda^{\ast}, \bds^{\ast}; \bdm) = \min_{(\bdlambda , \bds) \in \Xspace}L(\bdlambda, \bds; \bdm),
\end{equation*}
which shows that $(\bdlambda^{*},\bds^{*})$ is a solution of \cref{eq:min_L_X}. The proof is complete. 
\end{proof}

\begin{theorem}
\label{unique_theorem}
Suppose that the assumptions in \cref{Analemma2} hold, and there exists a point $(\bdlambda^{\diamond}, \bds^{\diamond})\in \Xspace$ such that
\begin{equation}\label{eq:solution_consist}
    m_{it} = \exp{(-s_i^{\diamond})}\sum_j a_{ijt}\lambda_j^{\diamond} \quad \forall (i, t), 
\end{equation}
and furthermore, for each LOR $i$, $A_{i} = (a_{ijt})_{J\times T}$ is full row rank. Then the uniqueness of the solution to \cref{eq:min_L_X} can be determined up to a constant.
\end{theorem}
\begin{proof}
\label{Anaproofthm2}
Based on \cref{neg_log_likelihood}, we consider the function with respect to $\tilde{\bdm} = \{\tilde{m}_{it}\}_{i, t} > 0$ as follows 
\begin{equation}
    \label{reduced_negetive_function}
    \tilde{L}(\tilde{\bdm}; \bdm) = \sum_{i,t}\tilde{m}_{it} - m_{it}\ln{\tilde{m}_{it}}. 
\end{equation}
Since the function $\tilde{L}(\Cdot; \bdm)$ in \cref{reduced_negetive_function} is strictly convex with respect to $\tilde{\bdm}$, its minimizer is uniquely determined, and denoted by $\tilde{\bdm}^{\diamond} = \bdm$. Then, combining with \cref{eq:solution_consist}, we get 
\begin{equation*}
        L(\bdlambda^{\diamond}, \bds^{\diamond}; \bdm) = \tilde{L}(\tilde{\bdm}^{\diamond}; \bdm) \le \tilde{L}(\tilde{\bdm}; \bdm) \quad \forall \tilde{\bdm} > 0.
\end{equation*} 
Particularly, 
\begin{equation*}
    L(\bdlambda^{\diamond}, \bds^{\diamond}; \bdm) \le \tilde{L}(\bar{\bdm}; \bdm) = L(\bdlambda, \bds; \bdm) \quad\forall (\bdlambda,\bds)\in \Xspace.
\end{equation*} 
Hence, $(\bdlambda^{\diamond},\bds^{\diamond})$ is a solution to \cref{eq:min_L_X}. 

If there exists another solution to \cref{eq:min_L_X}, denoted by $(\bdlambda^{\triangleright},\bds^{\triangleright})$, then 
\[
L(\bdlambda^{\triangleright}, \bds^{\triangleright}; \bdm) = L(\bdlambda^{\diamond}, \bds^{\diamond}; \bdm). 
\]
Let $\tilde{\bdm}^{\triangleright} = \{\tilde{m}_{it}^{\triangleright}\}_{i,t}$, where 
\[
\tilde{m}_{it}^{\triangleright} = \exp{(-s_i^{\triangleright})}\sum_j a_{ijt}\lambda_j^{\triangleright} \quad \forall (i, t). 
\]
Then, we have
\[
\tilde{L}(\tilde{\bdm}^{\triangleright}; \bdm) = \tilde{L}(\tilde{\bdm}^{\diamond}; \bdm). 
\]
By the uniqueness of the minimizer of the function $\tilde{L}(\Cdot; \bdm)$ in \cref{reduced_negetive_function}, we know 
\[
\tilde{\bdm}^{\triangleright} = \tilde{\bdm}^{\diamond}. 
\]
Equivalently, 
\[
\exp{(-s_i^{\triangleright})}\sum_j a_{ijt}\lambda_j^{\triangleright}  = \exp{(-s_i^{\diamond})}\sum_j a_{ijt}\lambda_j^{\diamond} \quad \forall (i, t).
\]
We can rewrite the equalities above with regard to $t$ in a compact form into 
\begin{equation*}
       \exp{(-s_i^{\triangleright})} A_i^{\tra} \bdlambda^{\triangleright} = \exp{(-s_i^{\diamond})} A_i^{\tra} \bdlambda^{\diamond} \quad \forall i. 
\end{equation*}
Thus, it follows that 
\begin{equation*}
    A_i^{\tra}\bigl(\bdlambda^{\triangleright} \exp{(-s_i^{\triangleright})} - \bdlambda^{\diamond} \exp{(-s_i^{\diamond})}\bigr) = 0\quad\forall i.
\end{equation*}
Thanks to $A_{i}$ being full row rank for each $i$, we immediately obtain 
\begin{equation}\label{Anathm2equ3}
    \bdlambda^{\triangleright} \exp{(-s_i^{\triangleright})} = \bdlambda^{\diamond} \exp{(-s_i^{\diamond})} \quad \forall i.
\end{equation}
Letting $\gamma_i = s_{i}^{\triangleright} - s_{i}^{\diamond}$, from \cref{Anathm2equ3}, we have 
\begin{equation*}
    \bdlambda^{\triangleright} = \bdlambda^{\diamond} \exp{(\gamma_i)},\quad\forall i.
\end{equation*}
Finally, it must be fulfilled that $\gamma_{1} = \gamma_{2}=\cdots=\gamma_{I}$, and then there is a constant difference between $\bds^{\triangleright}$ and $\bds^{\diamond}$, and the difference between $\bdlambda^{\triangleright}$ and $\bdlambda^{\diamond}$ is a constant multiple of the previous constant interest. Hence, the minimization problem \cref{eq:min_L_X} is uniquely solvable up to a constant. The proof is complete.
\end{proof}
\begin{remark}
\label{remark1}
Note that the result of \cref{unique_theorem} can be viewed as a discrete counterpart of the continuous case in \cite{defrise2012}. To the best of our knowledge, this work is the first in the literature on the uniqueness for the simultaneous reconstruction problem based on the discrete negative log-likelihood. 
\end{remark}

Let 
\[
\Yspace = \{(\bdlambda, \bds) \in \Xspace | \bdOneMask^{\tra} \bdlambda = N\}.
\]
In what follows, we give the results on the existence and uniqueness of the solution to the proposed model \cref{equ1}-\cref{eq_constraint},
equivalently,  
\begin{equation}\label{eq:min_L_Y} 
\min_{(\bdlambda , \bds) \in \Yspace} L(\bdlambda, \bds; \bdm). 
\end{equation}  

\begin{lemma}\label{Analemma4}
Assume that $a_{max} > 0$, and the measured data $\bdm > 0$. The function $L(\Cdot, \Cdot; \bdm)$ of \cref{eq:min_L_Y} is lower bounded and coercive in $\Yspace$. 
\end{lemma}
\begin{proof}
The lower boundedness and coercivity can be immediately proved from the proofs of \cref{Analemma1,Analemma2}, respectively. So the detail is omitted. 
\end{proof}

\begin{theorem}
\label{existing_theorem_constrained}
Suppose that the assumptions in \cref{Analemma4} hold. The solution to \cref{eq:min_L_Y} exists.
\end{theorem}
\begin{proof}
\label{proof_existing theorem for constrained case}
Here the proof has no difference from that of \cref{existing_theorem}. Thus we omit it here.
\end{proof}

\begin{theorem}
\label{unique_theorem_constrained}
Let the assumptions in \cref{Analemma4} hold. Assume that there exists a point $(\bdlambda^{\triangleleft}, \bds^{\triangleleft})\in \Yspace$ such that
\begin{equation}\label{eq:solution_consist_2}
    m_{it} = \exp{(-s_i^{\triangleleft})}\sum_j a_{ijt}\lambda_j^{\triangleleft} \quad \forall (i, t), 
\end{equation}
and furthermore, for each LOR $i$, $A_{i} = (a_{ijt})_{J\times T}$ is full row rank. Then, the solution is unique for the problem \cref{eq:min_L_Y}, and it is that $(\bdlambda^{\triangleleft}, \bds^{\triangleleft})$. 
\end{theorem}
\begin{proof}
From the proof of \cref{unique_theorem}, we know that $(\bdlambda^{\triangleleft}, \bds^{\triangleleft})$ is a solution to \cref{eq:min_L_Y}, and different solutions differ by only one constant. Furthermore, due to the constraint of $\bdOneMask^{\tra} \bdlambda = N$ in \cref{eq_constraint}, the solution can be determined uniquely.   
\end{proof} 

\begin{remark}\label{rem:nonuniqueness}
Note that the constraint on the total amount of the activity in some mask region in \cref{eq_constraint} is used to overcome the non-uniqueness in some  existing models \cite{defrise2012,rezaei2012,rezaei2014}, which is also a weaker condition than that in the recent work \cite{ren2024}. 
\end{remark}

Moreover, we give the analysis for the stability of the solution to the problem \cref{eq:min_L_Y}.

\subsubsection{Stability}

Following the work of \cite{burger2014}, we also consider a class of perturbations on the measured data in the form of 
\begin{equation}
\label{stability_equ1}
    \bdm^n = \bdm + \bdvsig^n \quad \text{with}\quad  \|\bdvsig^n\| \longrightarrow 0 \quad \text{as}\quad n \longrightarrow 0,
\end{equation}
where $\bdm^{n} = \{m_{it}^n\}_{i,t}$, and $\boldsymbol{\varsigma}^n = \{\varsigma_{it}^n\}_{i,t}$.
Then, the corresponding minimization problem based on the perturbed data can be stated as
\begin{align}
	\label{stability_equ2}
		\nonumber\min_{(\bdlambda, \bds) \in \Yspace} L(\bdlambda, \bds; \bdm^n) = \sum_{i,t}\Bigg\{&\exp{(-s_i)}\sum_{j=1}^{J}a_{ijt}\lambda_j \\
		&- m_{it}^n\ln\Bigg(\exp{(-s_i)}\sum_{j=1}^{J}a_{ijt}\lambda_{j}\Bigg)\Bigg\}. 
\end{align}
\begin{lemma}
\label{stability_lemma1}
Suppose that for a fixed $(\bdlambda, \bds)\in \Yspace$, $\min_{i,t}\left\{\sum_{j = 1}^J a_{ijt}\lambda_j\right\}> 0$. Then
\[
\lim_{n\to \infty}L(\bdlambda, \bds; \bdm^n) = L(\bdlambda, \bds; \bdm). 
\]
\end{lemma}
\begin{proof}
\label{stability_lemma_proof}
For the fixed $(\bdlambda,\bds)$, we have  
\begin{align}\label{stability_equ3}
  \nonumber | L(\bdlambda, \bds; \bdm^n) - &L(\bdlambda,\bds; \bdm)| \\
 \nonumber  &= \Bigg{|}\sum_{i,t}(m_{it}^{n} - m_{it})\ln\left(\exp{(-s_i)}\sum_{j = 1}^J a_{ijt}\lambda_j\right)\Bigg{|}\\
   \nonumber &\le\sum_{i,t}|\varsigma_{it}^{n}| \Bigg{|}\ln\left(\exp{(-s_i)}\sum_{j = 1}^J a_{ijt}\lambda_j\right)\Bigg{|}\\
   &\le \sqrt{IT}\|\bdvsig^{n}\|\max_{i,t}\left\{\Bigg{|}\ln\left(\exp{(-s_i)}\sum_{j = 1}^J a_{ijt}\lambda_j\right)\Bigg{|}\right\}.
\end{align}
Because of the $\bds\ge 0$ and $0 \le \bdlambda \le \Lambda$, and $\min_{i,t}\left\{\sum_{j = 1}^J a_{ijt}\lambda_j\right\}> 0$, it follows that for any $(i,t)$, 
\begin{align}\label{eq:bound_barm}
    \nonumber 0 < \exp{\bigl(-\max_i\{s_i\}\bigr)}\min_{i,t}\left\{\sum_{j = 1}^J a_{ijt}\lambda_j\right\} &\le \exp{(-s_i)}\sum_{j = 1}^J a_{ijt}\lambda_j \\
    &\le J\Lambda a_{max} <\infty. 
\end{align}
Using \cref{stability_equ1} and \cref{eq:bound_barm}, from \cref{stability_equ3}, we obtain 
\begin{equation*}
    |L(\bdlambda, \bds; \bdm^n) - L(\bdlambda,\bds; \bdm)| \longrightarrow 0 \quad \text{as} ~ n\longrightarrow \infty. 
\end{equation*}
The proof is complete. 
\end{proof}

\begin{remark}\label{stability_remark1}
The assumption that $\min_{i,t}\left\{\sum_{j = 1}^J a_{ijt}\lambda_j\right\} > 0 $ for some $(\bdlambda, \bds)\in \Yspace$ is reasonable since each LOR contains activity. Then it also makes sure that the logarithmic functions in \cref{neg_log_likelihood} and \cref{stability_equ2} are well-defined. 
\end{remark}

Denote  
\[
{\bdu^{\ast}}^n := ({\bdlambda^{\ast}}^n, {\bds^{\ast}}^n) \in \argmin_{(\bdlambda, \bds) \in \Yspace} L(\bdlambda, \bds; \bdm^n), 
\]
and 
\[
\bdu^{\ast} := (\bdlambda^{\ast},\bds^{\ast}) \in \argmin_{(\bdlambda, \bds) \in \Yspace} L(\bdlambda, \bds; \bdm). 
\]
\begin{theorem}
Let the assumptions in \cref{unique_theorem_constrained} hold. Assume that there exist positive constants $C_1$ and $C_2$ such that 
\begin{equation*}
    \bdm^n \ge C_1\quad\text{and}\quad\min_{i,t}\left\{\sum_{j = 1}^J a_{ijt}{\lambda_j^{\ast}}^n\right\}>C_2 \quad \forall n,
\end{equation*}
and $\min_{i,t}\left\{\sum_{j = 1}^J a_{ijt}\lambda^{\ast}_j\right\}>0$. Then  
\begin{equation*}
   \lim_{n\to\infty}{\bdu^{\ast}}^n = \bdu^{\ast}.
\end{equation*}
\end{theorem}
\begin{proof}
Since ${\bdu^{\ast}}^n$ is a minimizer of $L(\bdlambda, \bds; \bdm^n)$, it follows that
\begin{equation}\label{eq:u_star_n_u}
     L({\bdu^{\ast}}^n; \bdm^n)\le L(\bdu^{\ast}; \bdm^n). 
\end{equation}
Using \cref{eq:u_star_n_u}, and the assumption $\min_{i,t}\left\{\sum_{j = 1}^J a_{ijt}\lambda^{\ast}_j\right\}>0$, by \cref{stability_lemma1}, we have 
\begin{equation}\label{eq:L_bounded}
    \limsup_{n\to\infty}L({\bdu^{\ast}}^n; \bdm^n)\le\limsup_{n\to\infty}L(\bdu^{\ast}; \bdm^n) = L(\bdu^{\ast}; \bdm).
\end{equation}
By \cref{eq:L_bounded}, we know that the sequence $\{L({\bdu^{\ast}}^n; \bdm^n)\}_{n\in\mathbb{N}}$ is bounded. From the proof of \cref{Analemma2}, it is known that $\{{\bdu^{\ast}}^n\}_{n\in\mathbb{N}}$ is also bounded. Therefore, there exists a subsequence $\{{\bdu^{\ast}}^{n_k}\}_{k\in\mathbb{N}}$ and $\tilde{\bdu}^*\in Y$ such that
\begin{equation*}
    \lim_{k\to\infty} {\bdu^{\ast}}^{n_k} = \tilde{\bdu}^*. 
\end{equation*}
Using the continuity of $L(\Cdot,\Cdot; \bdm)$, \cref{stability_equ3}, and \cref{eq:L_bounded}, we obtain 
\begin{align}
   \nonumber L(\tilde{\bdu}^*; \bdm) &= \lim_{k\to\infty}L({\bdu^{\ast}}^{n_k}; \bdm) = \limsup_{k\to\infty}L({\bdu^{\ast}}^{n_k}; \bdm)\\
   \nonumber &\le \limsup_{k\to\infty}\bigl(L({\bdu^{\ast}}^{n_k}; \bdm)-L({\bdu^{\ast}}^{n_k}; \bdm^{n_k})\bigr) + \limsup_{k\to\infty}L({\bdu^{\ast}}^{n_k}; \bdm^{n_k}\})\\
  \nonumber \label{eq:L_v_L_u} &\le \limsup_{k\to\infty}\sqrt{IT}\|\bdvsig^{n_k}\|\max_{i,t}\left\{\Bigg{|}\ln\left(\exp{\bigl(-{s_i^{\ast}}^{n_k}\bigr)} \sum_{j=1}^J a_{ijt} {\lambda_j^{\ast}}^{n_k} \right)\Bigg{|}\right\} \\
   &\hspace{4mm}+L(\bdu^{\ast}; \bdm).
\end{align}
Combining the boundedness of $\{{\bdu^{\ast}}^n\}_{n\in\mathbb{N}}$ and the assumption on \\ $\min_{i,t}\left\{\sum_{j = 1}^J a_{ijt}{\lambda_j^{\ast}}^n\right\}$ for all $n$, we have 
\begin{equation}\label{eq:u_n_m_bounded}
0<C_3\le \exp{\bigl(-{s_i^{\ast}}^n\bigr)} \sum_{j=1}^J a_{ijt}{\lambda_j^{\ast}}^n \le J\Lambda a_{max} \quad \forall (i,t)~\text{and}~\forall n \in \mathbb{N}.
\end{equation}
With \cref{stability_equ1} and \cref{eq:u_n_m_bounded}, it is easy to obtain that the first term on the right side of \cref{eq:L_v_L_u} is vanishing. 
Thus, we have $L(\tilde{\bdu}^*; \bdm) \le L(\bdu^*; \bdm)$ immediately. Inversely, since $\bdu^*$ is the minimizer of $L(\Cdot, \Cdot; \bdm)$, it follows that $L(\tilde{\bdu}^*; \bdm) \ge L(\bdu^*; \bdm)$. Then 
\begin{equation*}
    L(\tilde{\bdu}^*; \bdm) = L(\bdu^*; \bdm).
\end{equation*}
By \cref{unique_theorem_constrained}, we know
that $\tilde{\bdu}^* = \bdu^*$. Finally, we have 
\begin{equation*}
   \lim\limits_{n\to\infty}{\bdu^{\ast}}^n=\bdu^*. 
\end{equation*}
The proof is complete. 
\end{proof}

\section{Computational Method}
\label{sec:method}

To solve the proposed model \cref{equ1}-\cref{eq_constraint}, we develop an intertwined update algorithm in \cref{alg:intertwined_alg}. 
\begin{algorithm}[htbp]
\caption{The intertwined update algorithm to solve the proposed model \cref{equ1}-\cref{eq_constraint}.}
\label{alg:intertwined_alg}
\begin{algorithmic}[1]
\STATE \emph{Initialize}: Given the measured data $\bdm$, and initialized the activity image $\bdlambda^0$ and the attenuation sinogram $\bds^0$. Let $k \gets 0$. 

\STATE \emph{Loop}:
\STATE\quad Solve the following optimization problem with the initialization $\bdlambda^k$: 
\begin{equation}
\label{methodequ1}
    \bdlambda^{k+1} = \argmin_{\bdlambda \in  \Yspace_1} L(\bdlambda,\bds^{k}; \bdm), 
\end{equation}
\newline \indent \hspace{2mm}  where $\Yspace_1 = \{\bdlambda \in \mathbb{R}^J | 0 \le \bdlambda \le \Lambda, \bdOneMask^{\tra}\bdlambda = N\}$. 

\STATE \quad Solve the following optimization problem with the initialization $\bds^k$: 
\begin{equation}
\label{methodequ2}
    \bds^{k+1} = \argmin_{\bds \in \Yspace_2}L(\bdlambda^{k+1},\bds; \bdm), 
\end{equation}
\newline \indent \hspace{2mm}  where $\Yspace_2 = \{\bds \in \mathbb{R}^I | \bds \ge 0\}$. 

\STATE \quad If some given termination condition is satisfied, \textbf{output} $\bdlambda^{k+1}$ and $\bds^{k+1}$; 
\newline \indent \hspace{2mm} Otherwise, let $k \gets k+1$, \textbf{goto} \emph{Loop}.
\end{algorithmic}
\end{algorithm} 

\subsection{The iterative scheme}
Based on \cref{alg:intertwined_alg}, the iterative scheme that we devise is given by
\begin{align}
	\label{equ2_general}
	\lambda_{j}^{k+1} &= \frac{\lambda_{j}^{k}}{\sum_{i,t}a_{ijt}\exp(-s^k_i)}\sum_{i,t}\left\{a_{ijt}\left(\frac{m_{it}}{\sum_{l}a_{ilt}\lambda_{l}^{k}}\right)\right\} \quad\forall~j,\\
	\label{equ3_general}
	\bdlambda^{k+1} &\longleftarrow \bdlambda^{k+1}\left(\frac{N}{\bdOneMask^{\tra}\bdlambda^{k+1}}\right),\\
	\label{equ4_general}
	s_{i}^{k+1} &= \proj\left[\ln\left(\frac{\sum_{j,t}a_{ijt}\lambda_{j}^{k+1}}{\sum_{t}m_{it}}\right)\right]\quad\forall~ i,
\end{align} 
where $\proj[\,\cdot\,]$ denotes the projection operator onto the set of nonnegative real numbers, and the scheme begins with the initial selection of $\lambda_{j}^{0}>0$ for all $j$ and $s_{i}^{0}\ge 0$ for all $i$. 

Since the subproblem \cref{methodequ1} has no closed-form solution, the first two steps of \cref{equ2_general} and \cref{equ3_general} are used to solve it iteratively. In contrast, the last step of \cref{equ4_general} can be used to solve the subproblem \cref{methodequ2} accurately. The detailed derivation for the iterative scheme is provided in \cref{sec:appendixA}.

\subsection{Convergence analysis}\label{convergence_analysis}

Here we analyze the convergence of the proposed \cref{alg:intertwined_alg}. Firstly, we give the following result on the solvability of the subproblems of \cref{methodequ1} and \cref{methodequ2}. 

\begin{lemma}
\label{prop:solvability}
Assume that the measured data $\bdm > 0$. The problems in \cref{methodequ1} and \cref{methodequ2} are solvable. 
Furthermore, if $A_{i} = (a_{ijt})_{J\times T}$ is full row rank for each LOR $i$, they are uniquely solvable. 
\end{lemma}

\begin{proof}
By \cref{neg_log_likelihood}, fixing $\bds =\bds^k$, we derive that 
\begin{equation*}
    \frac{\partial^{2} L(\bdlambda,\bds^{k}; \bdm)}{\partial\lambda_{j}\partial\lambda_{l}} = \sum_{i,t}m_{it}\frac{a_{ijt}a_{ilt}}{(\sum_{\xi}a_{i\xi t}\lambda_{\xi})^2}.
\end{equation*}
For any $0\neq\bz\in\mathbb{R}^{J}$, it follows that
\begin{equation*}
    \begin{aligned}
       \bz^{\tra}\left(\frac{\partial^{2} L(\bdlambda,\bds^{k}; \bdm)}{\partial\lambda_{j}\partial\lambda_{l}}\right)_{J\times J}\bz &= \sum_{j}\sum_{l}z_{j}\sum_{i,t}m_{it}\frac{a_{ijt}a_{ilt}}{(\sum_{\xi}a_{i\xi t}\lambda_{\xi})^2}z_{l}\\
       & = \sum_{i,t}m_{it}\left(\frac{\sum_{j}a_{ijt}z_{j}}{\sum_{\xi}a_{i\xi t}\lambda_{\xi}}\right)^{2} \ge 0. 
    \end{aligned}
\end{equation*}
Hence, the problem in \cref{methodequ1} is convex, which is solvable. 

In addition, fixing $\bdlambda =\bdlambda^{k+1}$, we have that 
\begin{align}
    \label{convergence_2nd derivative}
    \frac{\partial^{2} L(\bdlambda^{k+1},\bds; \bdm)}{\partial s_{i}\partial s_{l}} = 
    \left\{
    \begin{aligned}
        &\sum_{t}\left\{\exp{(-s_i)}\sum_{j} a_{ijt}\lambda_{j}^{k+1} \right\}\ge 0, \quad i = l,\\[2mm]
        & 0 ,\quad i \neq l. 
    \end{aligned}\right.
\end{align}
So the problem in \cref{methodequ2} is convex, which is also solvable. 

Moreover, if $A_{i} = (a_{ijt})_{J\times T}$ is full row rank for each LOR $i$, the corresponding Hessian matrix for each problem is positive definite. Hence, the uniqueness of the solution is achieved. 
\end{proof}

From the proof of \cref{prop:solvability}, the optimization problems in \cref{methodequ1} and \cref{methodequ2} are convex. Hence, \cref{alg:intertwined_alg} is referred to as the alternate convex search (ACS) \cite{gorski2007}, and also a special case of the block-relaxation methods \cite{de1994}.

\begin{theorem}
\label{convergence_thm1}
Let the assumption in \cref{Analemma4} hold, and $\{\bdu^{k} = (\bdlambda^{k},\bds^{k})\}_{k\in\mathbb{N}}$ be a sequence generated by \cref{alg:intertwined_alg}. Then $\{L(\bdu^{k}; \bdm)\}_{k\in\mathbb{N}}$ is monotonically convergent. Furthermore, let $A_{i} = (a_{ijt})_{J\times T}$ be full row rank for each LOR $i$. If $\bdu^{k+1} \neq \bdu^{k}$, then 
\begin{equation*}
    L(\bdu^{k+1}; \bdm)<L(\bdu^{k}; \bdm). 
\end{equation*}
\end{theorem}
\begin{proof}
\label{convergence_proof1}
By \cref{prop:solvability}, the problems in \cref{methodequ1} and \cref{methodequ2} are both solvable, which means that 
\begin{equation*}
    L(\bdlambda^{k},\bds^{k}; \bdm)\ge L(\bdlambda^{k+1},\bds^{k}; \bdm)\ge L(\bdlambda^{k+1},\bds^{k+1}; \bdm),
\end{equation*}
namely, $L(\bdu^{k}; \bdm)\ge L(\bdu^{k+1}; \bdm)$. So $\{L(\bdu^{k}; \bdm)\}_{k\in\mathbb{N}}$ is a monotonically decreasing sequence. By \cref{Analemma4}, the sequence $\{L(\bdu^{k}; \bdm)\}_{k\in\mathbb{N}}$ is lower bounded. Hence, the sequence is convergent. 

Moreover, if $A_{i} = (a_{ijt})_{J\times T}$ is full row rank for each LOR $i$, by \cref{prop:solvability}, the problems in \cref{methodequ1} and \cref{methodequ2} are uniquely solvable. Using the method of proof by contradiction, it is easy to obtain the desired result.  
\end{proof}

To proceed, we introduce some required definitions (see \cite{bazaraa2006,gorski2007,xu2013}). 
\begin{definition}
\label{convergence_def1}
Let $\Omega\subseteq X\times Y\subseteq \mathbb{R}^{m}\times\mathbb{R}^{r}$ be a non-empty set, and $h: \Omega\to \mathbb{R}$ be a given function. We call $(\bx^{*}, \by^{*}) \in \Omega$ a Nash equilibrium point if 
\begin{equation*}
    \begin{aligned}
       h(\bx^{*},\by^{*})\le h(\bx,\by^{*}) \quad\forall\bx\in \Omega_{\by^{*}},\\
       h(\bx^{*},\by^{*})\le h(\bx^{*},\by) \quad\forall\by\in \Omega_{\bx^{*}},
    \end{aligned}
\end{equation*}
where $\Omega_{\by^{*}} = \{\bx\in X | (\bx,\by^{*})\in \Omega\}$, and $\Omega_{\bx^{*}} = \{\by\in Y | (\bx^{*},\by)\in \Omega\}$.
\end{definition}


\begin{definition}
\label{convergence_def2}
Let $\Omega\subseteq X\times Y\subseteq \mathbb{R}^{m}\times\mathbb{R}^{r}$ be a non-empty set, and $h: \Omega\to \mathbb{R}$ be a given function. Let $\bdu_{q}=(\bx_{q},\by_{q})\in \Omega$ with $q=1,2$. The $G$ is a set-value map from $\Omega$ onto its power set $\mathcal{P}(\Omega)$, which is defined by $\bdu_{2}\in G(\bdu_{1})$ if and only if
\begin{equation*}
    \begin{aligned}
       h(\bx_2, \by_{1})\le h(\bx,\by_{1}) \quad\forall\bx\in \Omega_{\by_{1}},\\
       h(\bx_{2}, \by_{2})\le h(\bx_{2},\by) \quad\forall\by\in \Omega_{\bx_{2}}.
    \end{aligned}
\end{equation*}
Then we call $G$ the algorithmic map of the ACS algorithm.
\end{definition}

Next, we provide a necessary condition for the convergence of the sequence yielded by \cref{alg:intertwined_alg}.
\begin{theorem}
\label{convergence_lemma2}
Assume that $\{\bdu^{k}\}_{k\in\mathbb{N}}$ is a sequence generated by \cref{alg:intertwined_alg} and convergent to $\bdu^*$. The $\bdu^*$ is a Nash equilibrium point.
\end{theorem}
\begin{proof}
The feasible domain of the proposed model \cref{equ1}-\cref{eq_constraint} is closed, and the objective function is continuous. By the proof of theorem 4.7 in \cite{gorski2007}, the desired result is obtained. 
\end{proof}

\begin{theorem}
\label{convergence_thm2}
Let the assumption in \cref{Analemma4} hold, and $\{\bdu^{k}\}_{k\in\mathbb{N}}$ be a sequence generated by \cref{alg:intertwined_alg}. The sequence has at least one accumulation point. Furthermore, if $A_{i} = (a_{ijt})_{J\times T}$ is full row rank for each LOR $i$, then all the accumulation points are Nash equilibrium points and have the same function value, and 
\[
\lim_{k\to\infty}\|\bdu^{k+1} - \bdu^{k}\| = 0. 
\]
Moreover, all of the points form a connected, compact set, and those located in the interior of the feasible domain are stationary points of the function $L(\Cdot, \Cdot; \bdm)$.
\end{theorem}
\begin{proof}
By \cref{convergence_thm1}, $\{L(\bdu^{k}; \bdm)\}_{k\in\mathbb{N}}$ is monotonically convergent. Combined with \cref{Analemma4}, it follows that the sequence $\{\bdu^{k}\}_{k\in\mathbb{N}}$ is bounded. Hence, The sequence has at least one accumulation point. 

Since $A_{i} = (a_{ijt})_{J\times T}$ is full row rank for each LOR $i$, by \cref{prop:solvability}, the problems in \cref{methodequ1} and \cref{methodequ2} are uniquely solvable. Then, using the proofs of theorem 4.9 and corollary 4.10 in \cite{gorski2007}, we obtain the desired results.
\end{proof}

\begin{remark}\label{rem:converg}
Under the assumptions of \cref{convergence_thm2}, if, furthermore, assume that the set of the Nash equilibrium points is composed of uniformly isolated points, then the sequence $\{\bdu^{k}\}_{k\in\mathbb{N}}$ admits a limit point in the set. Similar proof can be referred to \cite{xu2013}. 
\end{remark}

\section{Relationship to existing models}
\label{sec:comparison}

Here we establish the relationship between the proposed MLAAS and some existing algorithms, mainly including the maximum likelihood estimation for simultaneously reconstructing activity and attenuation (MLAA) \cite{rezaei2012}, and the maximum likelihood estimation for  simultaneously reconstructing activity and attenuation correction factors (MLACF) \cite{rezaei2014}. 

In MLAA, the expected measurement for $(i, t)$ is written as
\begin{equation}
\label{SAAmodel}
    \Bar{m}_{it} = \exp{\left(-\sum_{j=1}^J b_{ij}\mu_j\right)} \sum_{j=1}^{J}a_{ijt}\lambda_{j},\quad i=1,\cdots,I, ~t=1,\cdots, T, 
\end{equation}
where $\sum_{j=1}^J b_{ij}\mu_j$ is computed as the sinogram element of the attenuation image $\bmu = [\mu_1, \cdots, \mu_J]$ on LOR $i$, the $b_{ij}$ is the intersection length of LOR $i$ with pixel/voxel $j$ of $\bmu$, and $\exp{\left(-\sum_{j=1}^J b_{ij}\mu_j\right)}$ is acted as the associated attenuation correction factor. 

Since the data is Poisson distributed, the negative log-likelihood function can be formulated as 
\begin{align}
    \label{neg_log_lik_SAA}
   \nonumber L_{SAA}(\bdlambda,\bmu; \bdm) = \sum_{i,t}\Bigg{\{}&\exp{\left(-\sum_{j=1}^J b_{ij}\mu_j\right)} \sum_{j=1}^Ja_{ijt}\lambda_j \\
    &- m_{it}\ln\left(\exp{\left(-\sum_{j=1}^J b_{ij}\mu_j\right)}\sum_{j=1}^J a_{ijt}\lambda_j\right)\Bigg{\}}. 
\end{align}
Indeed, the aim of MLAA is to reconstruct the activity map $\bdlambda$ and the attenuation map $\bmu$ simultaneously from the data $\bdm$ by solving the following minimization problem: 
\begin{equation}\label{eq:SAA} 
\min_{(\bdlambda, \bmu) \in \mathbb{R}_+^J\times\mathbb{R}_+^J} L_{SAA}(\bdlambda, \bmu; \bdm). 
\end{equation}
The corresponding solving algorithm is presented by  
\begin{align*}
    &f_i^k = \exp{\left(-\sum_j b_{ij}\mu_j^k\right)} \quad \forall i, 
    \\
    &\lambda_j^{k+1} = \frac{\lambda_{j}^{k}}{\sum_{i,t}a_{ijt}f_{i}^{k}}\sum_{i,t}a_{ijt}\left(\frac{m_{it}}{\sum_{l}a_{ilt}\lambda_{l}^{k}}\right) \quad \forall j, \\
    &\psi_i^k = f_{i}^{k}\sum_{j,t}a_{ijt}\lambda_j^{k+1}\quad  \forall i, \\
    &\mu_j^{k+1} = \mu_j^k +\frac{\sum_i b_{ij}(\psi_i^k - \sum_t m_{it})}{\sum_i b_{ij}\psi_i^k\sum_l b_{il}} \quad \forall j. 
\end{align*}

In MLACF, the expected count for $(i, t)$ is written as  
\begin{equation}
\label{SAAmodel_2}
    \Bar{m}_{it} = f_i\sum_{j=1}^J a_{ijt}\lambda_j, \quad~i=1, \cdots, I,~t=1, \cdots, T, 
\end{equation}
where $f_i$ denotes the attenuation correction factor of LOR $i$. 

Likewise, as the data is Poisson distributed, the related negative log-likelihood can be formulated as 
\begin{equation}
    \label{neg_log_lik_SAACF}
    L_{SACF}(\bdlambda,\bdf; \bdm) = \sum_{i,t}\left\{f_i\sum_{j=1}^J a_{ijt}\lambda_j - m_{it}\ln\left(f_i\sum_{j=1}^J a_{ijt}\lambda_j\right)\right\}, 
\end{equation}
where $\bdf = [f_1, \ldots, f_I]^{\tra}$. 

The aim of MLACF is to estimate the activity map $\bdlambda$ and the attenuation correction factors $\bdf$ simultaneously from the data $\bdm$ by solving the following minimization problem: 
\begin{equation}\label{eq:SAACF} 
\min_{(\bdlambda, \bdf) \in \mathbb{R}_+^J\times (0,1]^{I}} L_{SACF}(\bdlambda, \bdf; \bdm). 
\end{equation}
The solving algorithm is given by  
\begin{align*}
    &\lambda_j^{k+1} = \frac{\lambda_j^k}{\sum_{i,t}a_{ijt}f_i^k}\sum_{i,t}a_{ijt}\left(\frac{m_{it}}{\sum_l a_{ilt}\lambda_l^{k}}\right) \quad \forall j\\
    &f_{i}^{k+1} = \frac{\sum_t m_{it}}{\sum_{j,t} a_{ijt}\lambda_j^{k+1}} \quad \forall i. 
\end{align*}

Note that the most significant difference between the proposed MLAAS and the existing methods of MLAA and MLACF is the mathematical form for  characterizing the attenuation correction factor as illustrated in \cref{mean_data}, \cref{SAAmodel}, and \cref{SAAmodel_2}. The proposed MLAAS takes advantage of the informed exponential form of the attenuation process as prior information compared to MLACF, which guarantees the rationality of mathematical modeling. Compared to MLAA, the proposed MLAAS directly uses the related attenuation sinogram rather than the X-ray transform of  the attenuation map, which is more accurate and easier to solve as indicated in the next section. Moreover, in model \cref{eq:SAACF} of MLACF, the feasible domain for the attenuation correction factor is not a closed set, which makes the problem more intractable to solve.

\section{Numerical experiments}
\label{sec:experiments} 

In this section, a specific example for the proposed model \cref{equ1}-\cref{eq_constraint} is considered by replacing the constraint \cref{eq_constraint}  and the upper bound in \cref{eq_constraint_1} with 
\begin{equation}\label{eq_constraint_spec}
\bdOne^{\tra}\bdlambda = N, 
\end{equation}
where the $\bdOne$ is a concrete selection of the $\bdOneMask$, resulting in the same constraint as in \cite{ren2024}. 

Based on \cref{equ2_general}-\cref{equ4_general}, the iterative scheme for solving the specific model is given by
\begin{align*}
	\lambda_{j}^{k+1} &= \frac{\lambda_{j}^{k}}{\exp(-s^k_i)\sum_{i,t}a_{ijt}}\sum_{i,t}\left\{a_{ijt}\left(\frac{m_{it}}{\sum_{l}a_{ilt}\lambda_{l}^{k}}\right)\right\} \quad\forall~j,\\
	\bdlambda^{k+1} &\longleftarrow \bdlambda^{k+1}\left(\frac{N}{\bdOne^{\tra}\bdlambda^{k+1}}\right),\\
	s_{i}^{k+1} &= \proj\left[\ln\left(\frac{\sum_{j,t}a_{ijt}\lambda_{j}^{k+1}}{\sum_{t}m_{it}}\right)\right]\quad\forall~ i. 
\end{align*}

To validate the performance of the proposed MLAAS, we conduct various numerical experiments using different phantoms such as XCAT phantom \cite{segars2010}, a synthetic phantom, and the digital reference object (DRO) \cite{pierce2015}. 

The compared algorithms include some existing methods such as MLAA, MLACF (see \cref{sec:comparison}). The following metrics of relative errors are used to evaluate the reconstruction quality quantitatively: 
\begin{equation*}
    \textrm{RE}_{\lambda}^{k} = \frac{\|\bdlambda^k-\bdlambda^*\|}{\|\bdlambda^*\|},\quad \textrm{RE}_{s}^{k} = \frac{\|\bds^k-\bds^*\|}{\|\bds^*\|} \quad\text{and}\quad \textrm{RE}_{m}^{k} = \frac{\|\bar{\bdm}^k - \bdm\|}{\|\bdm\|},
\end{equation*}
where $\bdlambda^k$ and $\bds^k$ denote the $k$-th iterated results of the activity and attenuation sinogram computed by the used method, $\bdlambda^*$ and $\bds^*$ the given truths, respectively, and $\bar{\bdm}^k$ is the vector with elements 
\[
\bar{m}_{it}^k = \exp(-s_{i}^{k})\sum_{j}a_{ijt}\lambda_{j}^{k}.
\] 
Note that $\textrm{RE}_{\lambda}^{k}$ and $\textrm{RE}_{s}^{k}$ stand for the relative errors of the $k$-th iterated results for the activity and attenuation sinogram, respectively, while $\textrm{RE}_{m}^{k}$ is the relative error for the data model. 

The numerical experiments are performed on a workstation running Python equipped with an Intel i9-13900K 3.0GHz CPU and Nvidia RTX A4000 GPU.

\subsection{Numerical convergence}\label{test1}

In this test, the XCAT phantom consisting of $256\times256$ image array with pixel size $0.117 \times 0.117\,\text{cm}^2$ originated from a 4D extended cardiac-torso  phantom (XCAT Version 2.0) \cite{segars2010}  is used to conduct the validation of numerical convergence for the proposed method MLAAS, as shown in the top of \cref{fig3_1_XCAT}.
\begin{figure}[htbp]
    \centering
    \includegraphics[width=0.9\linewidth,]{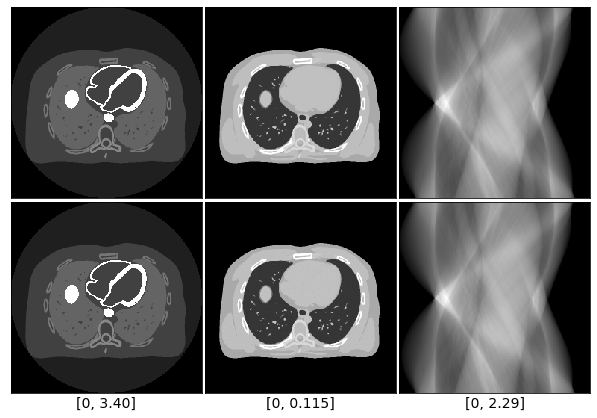}
    \caption{Top: The truths of the activity map (left), the attenuation map (middle) and the generated attenuation sinogram (right). Bottom: The reconstructed results of the activity map (left) and the attenuation sinogram (right) by using the proposed MLAAS, and accordingly the reconstructed attenuation image (middle) based on the obtained sinogram.}
    \label{fig3_1_XCAT}
\end{figure}

Field of view (FOV) is the inscribed circle of the square area, and parallel-beam scanning is performed. The scanning configuration for the phantom is including $256$ projection views evenly distributed over $[0,\pi)$, $256$ uniform detector bins with nearly $0.117\,\text{cm}$ bin-width at each view, and $10$ TOF bins with an effective TOF timing resolution of nearly $600\,\text{ps}$ for each LOR. The noise-free measured TOF sinogram data is $10\times256\times256$ in size. In addition, the  number of collected coincidence events is 1e6. 

Without loss of generality,  we initialize that $\bdlambda^0=1$ and $\bds^0=0$. In order to validate the numerical convergence of the proposed algorithm, the maximum iterative number is set to be a sufficiently large 5e4. After these iterations, the reconstructed results by the proposed method MLAAS are displayed in the bottom of \cref{fig3_1_XCAT}. Obviously, the reconstructed results of the activity and the attenuation sinogram are quite close to their truths, which demonstrates the good performance of the proposed method. It is worth noting that the reconstructed attenuation map is not a direct result of MLAAS, which is instead a by-product of inverting the estimated attenuation sinogram by some reconstruction method such as EM algorithm.  

Moreover, after sufficient iterations, we plot the curves of the metrics $\textrm{RE}_{\lambda}^{k}$, $\textrm{RE}_{s}^{k}$ and $\textrm{RE}_{m}^{k}$ as functions of iteration numbers in \cref{fig4}. It is shown that the proposed method has almost achieved numerical convergence and reconstructed the  accurate result, thereby demonstrating its effectiveness for inverting the noiseless data, which also confirms the convergence theory that established in \cref{convergence_analysis}. 
\begin{figure}[htbp]
    \centering
    \includegraphics[width=0.32\linewidth,]{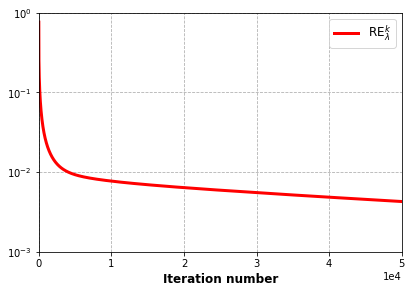}
    \includegraphics[width=0.32\linewidth,]{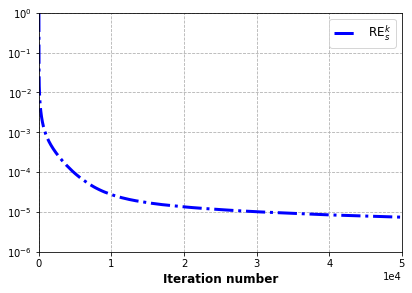}
    \includegraphics[width=0.32\linewidth,]{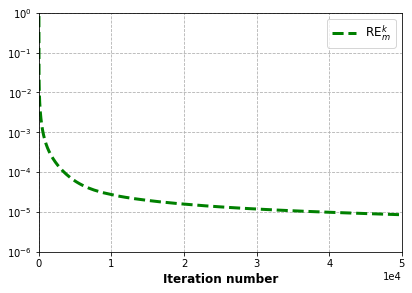}
    \caption{Metrics $\textrm{RE}_{\lambda}^{k}$, $\textrm{RE}_{s}^{k}$ and $\textrm{RE}_{m}^{k}$ are plotted in semi-log scale as functions of iteration numbers for simultaneously reconstructing the activity and the attenuation sinogram of the phantom given in \cref{fig3_1_XCAT} by the proposed algorithm MLAAS.}
    \label{fig4}
\end{figure}

\subsection{Numerical comparisons}
\label{test2}

To further evaluate the advantages of the proposed MLAAS, we make the comparison with several existing methods such as MLAA and MLACF (see \cref{sec:comparison}) by using three different tests.

\subsubsection{Test suite 1}

The used phantom is originated from the Digital Reference Object (DRO) of University of Washington \cite{pierce2015}, which is composed of $128\times128$ image array with the pixel size of $0.235\times0.235\,\text{cm}^2$ as shown in \cref{fig2_1}. The scanning configuration includes $128$ evenly parallel-beam projection views over $[0,\pi)$, $128$ uniform detector bins per view over $[-15,15]\,\text{cm}$, and $10$ TOF bins with an effective TOF timing resolution of nearly $600\,\text{ps}$ for each LOR. The measured coincidence events is set to be 1e6. Without loss of generality, we set the maximum iterative number to be 2e4 and the initialized $\bdlambda^0 = \boldsymbol{1}$ and $\bds^0 = \boldsymbol{0}$.
\begin{figure}[htbp]
    \centering
    \includegraphics[width=0.9\textwidth]{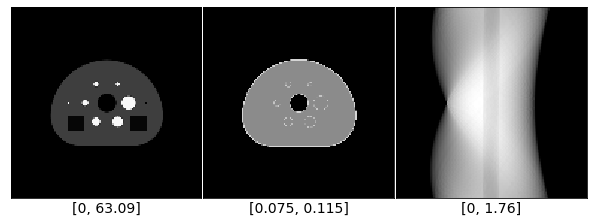}
    \caption{The truths of activity image (left), attenuation image (middle) and the corresponding attenuation sinogram (right) of the DRO phantom.}
    \label{fig2_1}
\end{figure}

In what follows, the activity and attenuation sinogram are reconstructed by utilizing MLAA, MLACF and the proposed MLAAS, respectively. Although the attenuation sinogram cannot be directly reconstructed by the first two methods, some additional process can be used to generate the corresponding attenuation sinogram. More precisely, Radon transform is performed to the obtained attenuation map in MLAA, and a negative logarithm operation is taken to the estimated attenuation correction factors in MLACF. The reconstructed results are displayed in \cref{fig2_2}.
\begin{figure}[htbp] 
   \centering 
    \begin{minipage}[htbp]{0.9\textwidth} 
   \centering 
   \subfloat[The activity images reconstructed by MLAA (left), MLACF (middle) and the proposed MLAAS (right) respectively.]{\includegraphics[width=\textwidth]{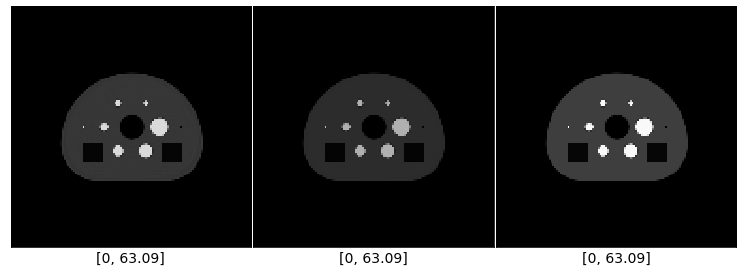}\label{fig2_2_1}} \\
   \subfloat[The attenuation sinograms reconstructed by MLAA (left), MLACF (middle) and the proposed MLAAS (right) respectively.]{\includegraphics[width=\textwidth]{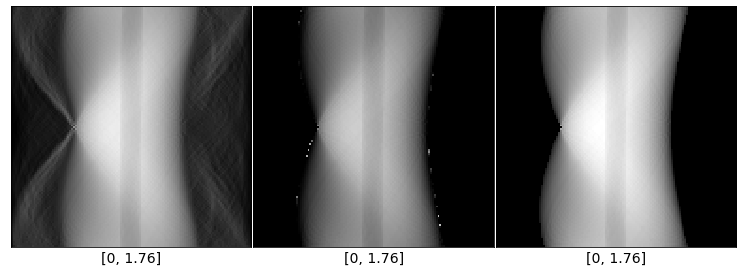}\label{fig2_2_2}} 
    \end{minipage} 
    \caption{The reconstructed results of activity image and attenuation sinogram after 2e4 iterations based on the DRO phantom.} 
    \label{fig2_2} 
\end{figure}

As shown in \cref{fig2_2_1}, the activity image reconstructed by MLAA suffers from moderate cross-talk and severe artifacts in the periphery of the scanned object, and that by MLACF sustains slight artifacts in the boundary of the imaged object. Instead, that by the proposed MLAAS is quite close to the truth. As for the estimated attenuation sinogram, it is clearly visible from \cref{fig2_2_2} that our algorithm can also generate better result than the other two ones. In addition, to compare the numerical performance, the curves of the relative errors $\textrm{RE}_{\lambda}^{k}$ and $\textrm{RE}_{m}^{k}$ as functions of iteration numbers are depicted in \cref{fig2_3}. As we can see, the proposed MLAAS has shown faster convergence rate than MLAA and MLACF.  Based on these comparisons, it is easy to observe that the proposed MLAAS has better performance of efficiency and accuracy than these two existing algorithms to handle the simultaneous estimation and further improves the final image quality.
\begin{figure}[H]
    \centering
    \includegraphics[width=0.49\linewidth,]{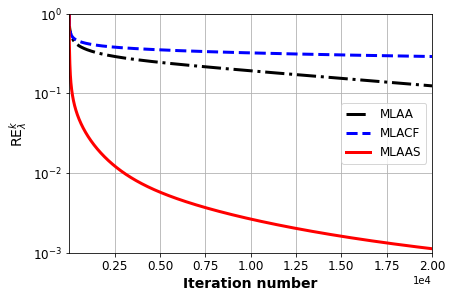}
    \includegraphics[width=0.49\linewidth,]{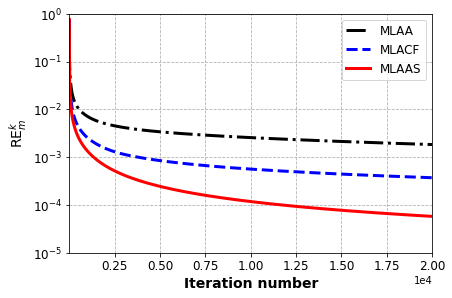}
    \caption{Metrics $\textrm{RE}_{\lambda}^{k}$ and $\textrm{RE}_{m}^{k}$ are plotted in a semi-log scale as functions of iteration numbers based on the DRO phantom. Note that the MLAAS is the proposed algorithm, and the MLAA and MLACF are two existing ones. }
    \label{fig2_3}
\end{figure}

\subsubsection{Test suite 2}
\label{test3}

To further illustrate, we design a synthetic phantom to validate the performance of the proposed MLAAS, where the phantom is discretized into   $64\!\times\!64$ pixels in a domain of $30\!\times\!30\,\text{cm}^2$, as shown in \cref{fig3_1}. 
\begin{figure}[htbp]
    \centering
    \includegraphics[width=0.9\linewidth,]{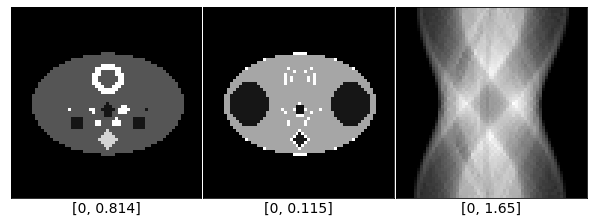}
    \caption{The truths of activity map (left), attenuation map (middle) and the corresponding attenuation sinogram (right) from the synthetic  phantom.}
    \label{fig3_1}
\end{figure}
The measured TOF-PET emission data are composed of $64$ parallel-beam projection views uniformly distributed over $[0,\pi)$, $64$ equidistant bins with about $0.469\,\text{cm}$ bin-width for each view and $10$ TOF bins for each LOR. The number of collected coincidence events is set to be 1e4.  Without loss of generality, we set the maximum iterative number to be 1e4, and initialize $\bdlambda^0 = \boldsymbol{1}$ and $\bds^0 = \boldsymbol{0}$.

The activity and attenuation sinogram are reconstructed by the proposed MLAAS, and the compared MLAA and MLACF respectively, which are shown in \cref{fig3_2}. 
\begin{figure}[htbp] 
   \centering 
    \begin{minipage}[htbp]{\textwidth} 
   \centering 
   \subfloat[The activity maps reconstructed by MLAA (left), MLACF (middle) and the proposed MLAAS (right) respectively.]{\label{fig3_2_1}\includegraphics[width=\textwidth]{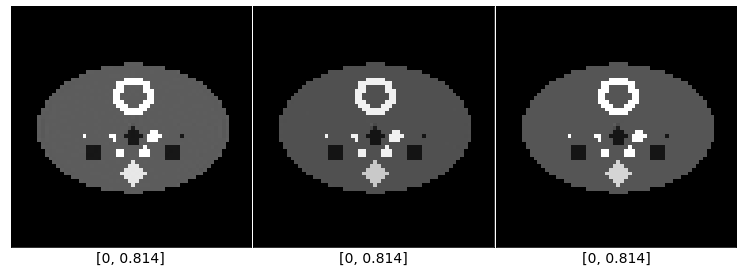}} \\
    \subfloat[The attenuation sinograms reconstructed by MLAA (middle), MLACF (middle) and the proposed MLAAS (right) respectively.]{\label{fig3_2_2}\includegraphics[width=\textwidth]{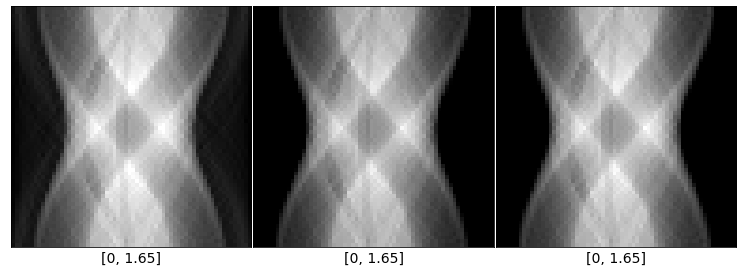}} 
    \end{minipage} 
    \caption{The reconstructed results of activity and attenuation sinogram after 1e4 iterations based on the synthetic phantom in \cref{fig3_1}.} 
    \label{fig3_2} 
\end{figure}
From a visual perspective, the images reconstructed by these three methods seem to be analogous. However, to assess the convergent behavior for these methods, we depict the curves of $\textrm{RE}_{\lambda}^{k}$ and $\textrm{RE}_{m}^{k}$ in \cref{fig3_3}. It is easy to figure out that the proposed MLAAS has shown faster convergence rate than MLAA and MLACF.
\begin{figure}[htbp]
    \centering
    \includegraphics[width=0.49\textwidth]{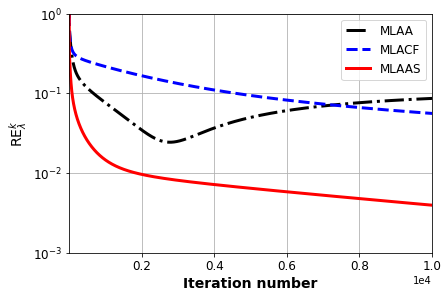}
    \includegraphics[width=0.49\textwidth]{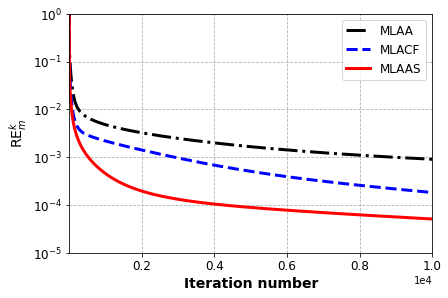}
    \caption{Metrics $\textrm{RE}_{\lambda}^{k}$ and $\textrm{RE}_{m}^{k}$ are plotted in a semi-log scale as functions of iteration numbers based on the synthetic phantom. Note that MLAAS is the proposed algorithm, and MLAA and MLACF are the two existing ones.}
    \label{fig3_3}
\end{figure}

Apart from the visual perception, the structural similarity (SSIM) and peak signal-to-noise (PSNR) are often used to evaluate the image quality quantitatively \cite{zhou2004}. So we compute the values of 1$-$SSIM and PSNR for these reconstructed results by different methods, as displayed in \cref{table3}. Note that the proposed method can achieve the best reconstruction quality among these three methods. 

\begin{table}[htbp]
\centering
\caption{The 1$-$SSIM and PSNR of the reconstructed activity and attenuation sinogram against the corresponding truths in the test suite 2. The $\downarrow$ (resp. $\uparrow$) indicates that a smaller (larger) value is better. Bold font means the best result. }
\label{table3}
\begin{tabular}{c *{2}{c} *{2}{c}} 
\toprule
&\multicolumn{2}{c}{Activity} & \multicolumn{2}{c}{Attenuation sinogram} \\
\cmidrule(lr){2-3} \cmidrule(lr){4-5} 
&1$-$SSIM\,$\downarrow$ & PSNR\,$\uparrow$ & 1$-$SSIM\,$\downarrow$ & PSNR\,$\uparrow$  \\
\cmidrule(lr){1-5}
MLAA & 3.91e-3 &33.70 &2.09e-1&23.14 \\
MLACF&1.74e-3&37.47 &7.37e-3&30.89\\
MLAAS (Ours) &\textbf{1.38e-5}&\textbf{60.50} &\textbf{1.18e-5}&\textbf{70.42}\\
\bottomrule[0.25mm]
\end{tabular}
\end{table}

\subsubsection{Test suite 3}
In practice, the measured data is frequently corrupted by noise. To simulate a more practical situation, we conduct the experiments by using the same phantom as in test suite 1 in  different noise levels. Note that the noise data is generated by a Poisson generator while the other settings remain unchanged unless otherwise stated. 
The datasets with low, moderate and high noise levels are generated, where the corresponding signal-to-noise ratios (SNRs) are about 27.23 dB, 17.21 dB and 7.25 dB, respectively. 

Since the noisy data is dealt with, the early-stopping technique is needed. Generally, the iteration in the experiment with higher noise level should be stopped earlier. For the low noisy case, the maximum iteration number is set to 1000. The activity images reconstructed by different methods are shown in \cref{fig5_1}.
\begin{figure}[htbp]
    \centering
    \includegraphics[width=0.9\linewidth,]{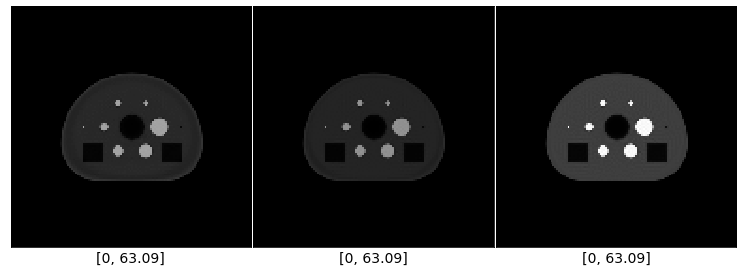}
    \caption{The activity images reconstructed by MLAA (left), MLACF (middle) and the proposed MLAAS (right) respectively from the 27.23 dB noisy data after 1000 iterations.}
    \label{fig5_1}
\end{figure}
Note that the activity image reconstructed by the proposed MLAAS is the closest to the truth. For quantitative comparison, we calculate the values of normalized root mean square error (NRMSE), SSIM and PSNR of the reconstructed activity map against the truth, which are tabulated in \cref{table7}. It shows that the proposed method can achieve the best performance among these three approaches.
\begin{table}[htbp]
	\centering
	\caption{The NRMSE, SSIM and PSNR of the reconstructed activity against the corresponding truth from the data in 27.23 dB noise level. The $\downarrow$ (resp. $\uparrow$) indicates that a smaller (larger) value is better. Bold font means the best result.}
	\label{table7}
	\begin{tabular}{cccc}
		\toprule[0.25mm]
		&NRMSE\,$\downarrow$& SSIM\,$\uparrow$& PSNR\,$\uparrow$ \\
		\midrule[0.1mm]
		MLAA &3.53e-1&0.9550&26.08 \\
		MLACF &4.22e-1&0.9378&24.53\\
		MLAAS (Ours) &\textbf{3.09e-2}&\textbf{0.9961}&\textbf{47.26}\\
		\bottomrule[0.25mm]
	\end{tabular}
\end{table} 

For the moderate noisy case, the maximum iteration number is set to 700. The computed results are shown in \cref{fig5_2}, and the quantitative indexes of NRMSE, SSIM and PSNR are tabulated \cref{table8}, which manifests that the value of NRMSE by MLAAS is smaller than those by MLAA and MLACF, and the values of SSIM and PSNR are larger than those by the other two methods. So, in this case, our method is also the best.
\begin{figure}[htbp]
    \centering
    \includegraphics[width=0.9\linewidth,]{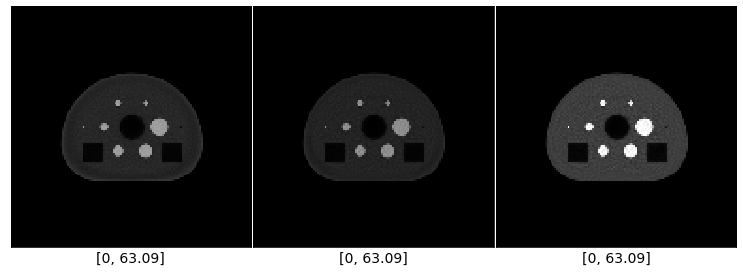}
    \caption{The activity maps reconstructed by MLAA (left), MLACF (middle) and the proposed MLAAS (right) respectively from the 17.21 dB noisy data after 700 iterations.}
    \label{fig5_2}
\end{figure}
\begin{table}
	\centering
	\caption{The NRMSE, SSIM and PSNR of the reconstructed activity against the corresponding truth from the data in 17.21 dB noise level. The $\downarrow$ (resp. $\uparrow$) indicates that a smaller (larger) value is better. Bold font means the best result.}
	\label{table8}
	\begin{tabular}{cccc}
		\toprule[0.25mm]
		&NRMSE\,$\downarrow$& SSIM\,$\uparrow$& PSNR\,$\uparrow$ \\
		\midrule[0.1mm]
		MLAA &3.77e-1&0.9434&25.52 \\
		
		MLACF &4.39e-1&0.9255&24.20\\
		
		MLAAS (Ours) &\textbf{7.12e-2}&\textbf{0.9746}&\textbf{40.00}\\
		\bottomrule[0.25mm]
	\end{tabular}
\end{table} 

Finally, for the high noisy case, the maximum iteration number is set to 51. The activity maps reconstructed by different methods are displayed in \cref{fig5_3}. Meanwhile, the quantitative indexes of NRMSE, SSIM and PSNR of the reconstructed activity against the corresponding truth are tabulated in \cref{table9}. Although the data noise is relatively high, the proposed MLAAS still outperforms MLAA and MLACF soundly. 
\begin{figure}[htbp]
    \centering
    \includegraphics[width=0.9\linewidth,]{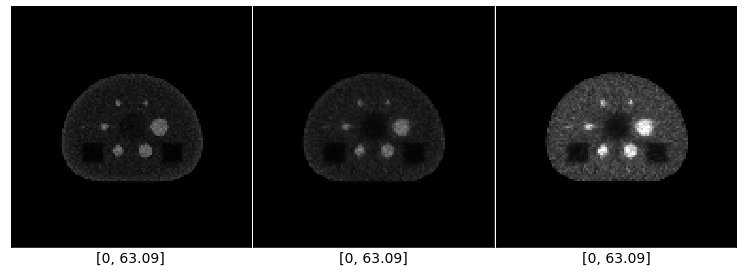}
    \caption{The activity images reconstructed by MLAA (left), MLACF (middle) and the proposed MLAAS (right) respectively from the 7.25 dB noisy data after 51 iterations.}
    \label{fig5_3}
\end{figure}

\begin{table}
	\centering
	\caption{The NRMSE, SSIM and PSNR of the reconstructed activity against the corresponding truth from the data in 7.25 dB noise level. The $\downarrow$ (resp. $\uparrow$) indicates that a smaller (larger) value is better. Bold font means the best result.}
	\label{table9}
	\begin{tabular}{cccc}
		\toprule[0.25mm]
		&NRMSE\,$\downarrow$& SSIM\,$\uparrow$& PSNR\,$\uparrow$ \\
		\midrule[0.1mm]
		MLAA &0.6041&0.8596&21.42 \\
		
		MLACF &0.5974&0.8538&21.52\\
		
		MLAAS (Ours) &\textbf{0.2526}&\textbf{0.8993}&\textbf{28.99}\\
		\bottomrule[0.25mm]
	\end{tabular}
\end{table} 

From these tests, even if in different noisy situations, it demonstrates that the proposed method remains robust to some extent, which makes it possible to deal with the data in practical applications.

\section{Discussion}
\label{discussion}

Over the last two decades, with the advent of scanners with TOF-capability, simultaneous reconstruction of activity and attenuation maps in TOF-PET has been extensively studied. That is because it has been shown that the use of TOF information can not only obtain more accurate activity than before but also alleviate the cross-talk problem. Moreover, in continuous case, previous works have demonstrated that there is a unique solution up to a constant for the joint estimation of activity and attenuation sinogram only from the TOF-PET emission data. In contrast, the current work developed a new ML approach to jointly reconstruct the activity and attenuation sinogram by additionally imposing a given total amount of the activity in some mask region of the FOV for the discrete case. Then, we also analyzed the existence, uniqueness and stability of the solution to the proposed model. In fact, estimating that amount in advance is a critical problem, and an alternative is to give an upper bound through the collected data, which is reserved for future research. 

In contrast to the standard MLAA algorithm in \cite{nuyts1999} that used to simultaneously recovery of activity map and attenuation map in PET, our method does not need to estimate the attenuation map but the attenuation sinogram. Note that the MLAA considers the situation where the effective region containing the injected radioactive tracer is the same as that suffering from attenuation. However, in practice, it does not always happen. Occasionally, some organ-specific tracer is used in the clinical application, which leads to serious inconsistence between the two mentioned regions. Although one attempt is to treat the attenuation reconstruction as a CT interior problem, it is still challenging. On the contrary, it can be overcome by the proposed MLAAS algorithm for which only needs to estimate the value of attenuation sinogram along each LOR. 


Compared to the MLACF algorithm that is employed to handle the simultaneous reconstruction of activity and attenuation correction factors, the proposed method can make use of the exponential property of the attenuation. On the one hand, it makes the algorithmic behaviours more stable and faster. On the other hand, it also avoids optimizing the associated problem over a semi open closed interval.

\section{Conclusion}
\label{conclusion}

In this article, we have devoted our efforts to study the approach with theoretical guarantees of simultaneously reconstructing the activity and attenuation sinogram for TOF-PET, which can achieve autonomous attenuation correction without the need for additional CT or MRI scan. 

Based on the maximum likelihood estimation, we proposed a new mathematical model for simultaneously reconstructing the activity and attenuation sinogram from the TOF-PET emission data only. In the proposed model, we made use of the exclusively exponential form for the attenuation correction factors, and also utilized the constraint of a total amount of the activity in some mask region. Under mild conditions, we proved that the proposed model is well-posed, the solution of which exists uniquely, and is also stable under a small perturbation of measurements. 

To solve the proposed model, we proposed an efficient algorithm to update the activity and attenuation sinogram alternatively. Under the condition that the two subproblems are exactly solved at each iteration, we obtained some useful convergent results for the algorithm. 

The performance of the proposed algorithm was assessed by a number of numerical experiments with noiseless/noisy emission data. It turns out that the proposed method, namely, MLAAS, is of numerical convergence and robust to noise, and has the capability of autonomous attenuation correction. Through numerical comparisons, it can be seen that the proposed MLAAS outperforms some state-of-the-art methods such as MLAA and MLACF in terms of accuracy and efficiency.

\appendix
\section{Appendix} 

\subsection{The derivation of the iterative scheme \cref{equ2_general}-\cref{equ4_general}}
\label{sec:appendixA}
To proceed, we rewrite \cref{neg_log_likelihood} of the objective function in \cref{equ1} as 
\begin{equation}
    \label{Appequ1}
    L(\bdlambda, \bds; \bdm) = \sum_{i,t} \left\{\exp{(-s_i)}\sum_{j=1}^{J}a_{ijt}\lambda_{j} - m_{it}\ln\left(\sum_{j=1}^{J}a_{ijt}\lambda_{j}\right) + m_{it}s_i \right\}. 
\end{equation}
Using the fact that $-\ln(\cdot)$ is a convex function, and the Jensen's inequality, for \cref{Appequ1}, we have  
\begin{align}
    \label{Appequ3}
       \nonumber L(\bdlambda, \bds; \bdm)  &\le \sum_{i,t}\Bigg\{\exp{(-s_i)}\sum_j a_{ijt}\lambda_{j} \\
       \nonumber &\hspace{5mm} - m_{it}\sum_j\frac{a_{ij t}\lambda_{j}^{k}}{\sum_{l}a_{ilt}\lambda_{l}^{k}}\ln\left(\frac{\sum_{l}a_{ilt}\lambda_{l}^{k}}{\lambda_{j}^{k}}\lambda_{j}\Bigg) +m_{it}{s_{i}}\right\}\\
        &:= Q(\bdlambda; \bds, \bdlambda^{k}).
\end{align}
From a perspective of non-statistical framework, the expectation step of EM algorithm is equivalent to replacing the original problem with another readily solvable problem, which is essentially a process of the optimization transfer \cite{de1995,erdogan1999,fessler1998,jacobson2003}. 

Fixed $\bds = \bds^{k}$ in \cref{Appequ3}, it is easy to verify that 
\begin{align*}
    Q(\bdlambda^k; \bds^k, \bdlambda^{k}) &= L(\bdlambda^k, \bds^k; \bdm),\\
    \nabla Q(\bdlambda; \bds^k, \bdlambda^{k})\big|_{\bdlambda = \bdlambda^{k}} &= \nabla L(\bdlambda, \bds^k; \bdm)\big|_{\bdlambda = \bdlambda^{k}},
\end{align*}
and  
\begin{equation}\label{Appequ4}
    Q(\bdlambda; \bds^k, \bdlambda^{k}) - Q(\bdlambda^k; \bds^k, \bdlambda^{k}) \ge L(\bdlambda, \bds^k; \bdm) - L(\bdlambda^k, \bds^k; \bdm),
\end{equation}
where $\nabla$ denotes the gradient operator. It is worth noting that the minimization of $Q(\bdlambda; \bds^{k}, \bdlambda^{k})$ is simple and uniquely solvable. Hence, $Q(\bdlambda; \bds^{k}, \bdlambda^{k})$ can be served as a surrogate function of $L(\bdlambda, \bds^k; \bdm)$. Then, if the  constraint is not considered, solving the subproblem of \cref{methodequ1} can be transferred into 
\begin{equation}
   \label{Appequ6}
   \boldsymbol{\lambda}^{k+1} = \argmin_{\bdlambda}Q(\bdlambda; \bds^{k}, \bdlambda^{k}).
\end{equation}
Taking the partial derivative for $Q(\bdlambda; \bds^{k}, \bdlambda^{k})$ with respect to $\lambda_{j}$, we immediately obtain that 
\begin{equation*}
    \frac{\partial Q(\bdlambda; \bds^{k}, \bdlambda^{k})}{\partial\lambda_{j}} = \sum_{i,t}\left\{\exp{(-s_{i}^{k})}a_{ijt} - m_{it}\frac{a_{ij t}\lambda_{j}^{k}}{\lambda_{j}\sum_{l}a_{ilt}\lambda_{l}^{k}}\right\}\quad\forall j.
\end{equation*}
Letting the above partial derivative be zero, we get the update of \cref{equ2_general}. 
Besides that, we also can prove that the solution is unique thanks to the fact that
\begin{equation*}
    \frac{\partial^{2} Q(\bdlambda; \bds^{k}, \bdlambda^{k})}{\partial\lambda_{j}\partial\lambda_{l}} = 
    \left\{
    \begin{aligned}
        &\sum_{i,t}\left\{m_{it}\frac{a_{ij t}\lambda_{j}^{k}}{(\lambda_{j})^{2}\sum_{l}a_{ilt}\lambda_{l}^{k}}\right\} > 0, \quad j = l,\\
        & 0 ,\quad j \neq l 
    \end{aligned}\right.
    \end{equation*}
for $\bdlambda \neq 0$. As to the constraint of the subproblem \cref{methodequ1}, we just need to perform the update \cref{equ3_general}.

Furthermore, using \cref{Appequ4} and \cref{Appequ6}, it follows that 
\begin{equation*}
    L(\bdlambda^{k+1}, \bds^k; \bdm) - L(\bdlambda^k, \bds^k; \bdm) \le Q(\bdlambda^{k+1}; \bds^k, \bdlambda^{k}) - Q(\bdlambda^k; \bds^k, \bdlambda^{k})\le0,
\end{equation*}
which guarantees that the value of the likelihood function gradually decreases with iterating
\begin{equation}
\label{Appfix_s}
    L(\bdlambda^{k+1}, \bds^k; \bdm) \le L(\bdlambda^k, \bds^k; \bdm).
\end{equation}

Similarly, letting $\bdlambda = \bdlambda^{k+1}$, we need to consider the subproblem of \cref{methodequ2}.  
Since the objective function in \cref{methodequ2} is separable for the attenuation sinogram on each LOR, the original minimization problem can be translated into solving the following problems independently 
\begin{equation}
    \label{Appequ10}
    s_{i}^{k+1} = \argmin_{s_{i} \ge 0}\sum_{t} \left\{\exp{(-s_i)}\sum_{j=1}^{J}a_{ijt}\lambda_{j}^{k+1} + m_{it}s_i \right\}\quad\forall~i. 
\end{equation}
Consequently, we can solve \cref{Appequ10} directly and then obtain the uniquely closed-form solution by \cref{equ4_general}.

\bibliographystyle{plain}
\bibliography{SAASreferences}

\begin{thebibliography}{10}

\bibitem{bazaraa2006}
M.~S. Bazaraa, H.~D. Sherali, and C.~M. Shetty.
\newblock {\em Nonlinear programming: theory and algorithms}.
\newblock John Wiley \& Sons, 2006.

\bibitem{bezrukov2013}
I.~Bezrukov, F.~Mantlik, H.~Schmidt, B.~Sch{\"o}lkopf, and B.~J. Pichler.
\newblock {MR-Based PET Attenuation Correction for PET/MR Imaging}.
\newblock {\em Seminars in Nuclear Medicine}, 43(1):45--59, 2013.

\bibitem{bronnikov1995}
A.~V. Bronnikov.
\newblock Approximate reconstruction of attenuation map in {SPECT} imaging.
\newblock {\em IEEE Trans. Nucl. Sci.}, 42(5):1483--1488, 1995.

\bibitem{bronnikov2000}
A.~V. Bronnikov.
\newblock Reconstruction of attenuation map using discrete consistency
  conditions.
\newblock {\em IEEE Trans. Med. Imaging}, 19(5):451--462, 2000.

\bibitem{burger2002}
C.~Burger, G.~Goerres, S.~Schoenes, A.~Buck, A.~Lonn, and G.~Von~Schulthess.
\newblock {PET} attenuation coefficients from {CT} images: experimental
  evaluation of the transformation of {CT} into {PET 511-keV} attenuation
  coefficients.
\newblock {\em Eur. J. Nucl. Med. Mol. Imaging}, 29(7):922--927, 2002.

\bibitem{burger2014}
M.~Burger, J.~Müller, E.~Papoutsellis, and C.~B. Schönlieb.
\newblock {Total variation regularization in measurement and image space for
  PET reconstruction}.
\newblock {\em Inverse Probl.}, 30(10):105003, 2014.

\bibitem{burgos2014}
N.~Burgos, M.~J. Cardoso, K.~Thielemans, M.~Modat, S.~Pedemonte, J.~Dickson,
  A.~Barnes, R.~Ahmed, C.~J. Mahoney, J.~M. Schott, J.~S. Duncan, D.~Atkinson,
  S.~R. Arridge, B.~F. Hutton, and S.~Ourselin.
\newblock Attenuation correction synthesis for hybrid {PET-MR} scanners:
  application to brain studies.
\newblock {\em IEEE Trans. Med. Imaging.}, 33(12):2332--2341, 2014.

\bibitem{censor1979}
Y.~Censor, D.~E. Gustafson, A.~Lent, and H.~Tuy.
\newblock A new approach to the emission computerized tomography problem:
  Simultaneous calculation of attenuation and activity coefficients.
\newblock {\em IEEE Trans. Nucl. Sci.}, 26(2):2775--2779, 1979.

\bibitem{clinthorne1991}
N.~H. Clinthorne, J.~A. Fessler, G.~D. Hutchins, and W.~L. Rogers.
\newblock {Joint maximum likelihood estimation of emission and attenuation
  densities in PET}.
\newblock In {\em Conf. Rec. IEEE Nucl. Sci. Symp. Med. Imaging Conf.},
  volume~3, pages 1927--1932, 1991.

\bibitem{conti2010}
M.~Conti.
\newblock Why is {TOF PET} reconstruction a more robust method in the presence
  of inconsistent data?
\newblock {\em Phys. Med. Biol.}, 56(1):155, 2010.

\bibitem{de1994}
J.~De~Leeuw.
\newblock Block-relaxation algorithms in statistics.
\newblock In {\em Information Systems and Data Analysis}, pages 308--324.
  Springer, 1994.

\bibitem{de1995}
A.~R. De~Pierro.
\newblock A modified expectation maximization algorithm for penalized
  likelihood estimation in emission tomography.
\newblock {\em IEEE Trans. Med. Imaging}, 14(1):132--137, 1995.

\bibitem{depierro2006}
A.~R. De~Pierro and F.~Crepaldi.
\newblock Activity and attenuation recovery from activity data only in emission
  computed tomography.
\newblock {\em Comput. Appl. Math.}, 25:205--227, 2006.

\bibitem{defrise2012}
M.~Defrise, A.~Rezaei, and J.~Nuyts.
\newblock Time-of-flight {PET} data determine the attenuation sinogram up to a
  constant.
\newblock {\em Phys. Med. Biol.}, 57(4):885--899, 2012.

\bibitem{erdogan1999}
H.~Erdogan and J.~A. Fessler.
\newblock Monotonic algorithms for transmission tomography.
\newblock {\em IEEE Trans. Med. Imaging}, 18(9):801--814, 1999.

\bibitem{fessler1998}
J.~A. Fessler and H.~Erdogan.
\newblock A paraboloidal surrogates algorithm for convergent
  penalized-likelihood emission image reconstruction.
\newblock In {\em IEEE Nuc. Sci. Symp. Med. Imag. Conf.}, volume~2, pages
  1132--1135, 1998.

\bibitem{gorski2007}
J.~Gorski, F.~Pfeuffer, and K.~Klamroth.
\newblock Biconvex sets and optimization with biconvex functions: a survey and
  extensions.
\newblock {\em Math. Methods Oper. Res}, 66(3):373--407, 2007.

\bibitem{helgason1965}
S.~Helgason.
\newblock The {Radon} transform on {Euclidean} spaces, compact two-point
  homogeneous spaces and {Grassmann} manifolds.
\newblock {\em Acta Math.}, 113:153--180, 1965.

\bibitem{jacobson2003}
M.~W. Jacobson and J.~A. Fessler.
\newblock Joint estimation of image and deformation parameters in
  motion-corrected {PET}.
\newblock In {\em 2003 IEEE Nuc. Sci. Symp. Conf. Rec. (IEEE Cat.)}, volume~5,
  pages 3290--3294, 2003.

\bibitem{karp2008}
J.~S. Karp, S.~Surti, Daube W., Margaret E., and G.~Muehllehner.
\newblock Benefit of time-of-flight in {PET}: experimental and clinical
  results.
\newblock {\em J. Nucl. Med.}, 49(3):462--470, 2008.

\bibitem{kinahan2003}
P.~E. Kinahan, B.~H. Hasegawa, and T.~Beyer.
\newblock X-ray-based attenuation correction for positron emission
  tomography/computed tomography scanners.
\newblock {\em Seminars in Nuclear Medicine}, 33(3):166--179, 2003.

\bibitem{krol1995}
A.~Krol, S.~H. Manglos, J.~F. Bowsher, T.~Young, D.~A. Bassano, and F.~D.
  Thomas.
\newblock Attenuation compensation in {SPECT} cardiac imaging using
  {EM-IntraSPECT} method.
\newblock {\em J. Nucl. Med.}, 36:50P, 1995.

\bibitem{lewellen1998}
Tom~K. Lewellen.
\newblock {Time-of-flight PET}.
\newblock {\em Seminars in Nuclear Medicine}, 28(3):268--275, 1998.

\bibitem{li2017joint}
Q.~Li, H.~Li, K.~Kim, and G.~El~Fakhri.
\newblock Joint estimation of activity image and attenuation sinogram using
  time-of-flight positron emission tomography data consistency condition
  filtering.
\newblock {\em J. Med. Imaging}, 4(2):023502, 2017.

\bibitem{ludwig1966}
D.~Ludwig.
\newblock The {Radon} transform on {Euclidean} spaces.
\newblock {\em Commun. Pure Appl. Math.}, 19:49--81, 1966.

\bibitem{moses2003}
W.~W. Moses.
\newblock Time of flight in {PET} revisited.
\newblock {\em IEEE Trans, Nucl. Sci.}, 50(5):1325--1330, 2003.

\bibitem{natterer1993}
F.~Natterer.
\newblock Determination of tissue attenuation in emission tomography of
  optically dense media.
\newblock {\em Inverse Probl.}, 9(6):731, 1993.

\bibitem{natterer1992}
F.~Natterer and H.~Herzog.
\newblock Attenuation correction in positron emission tomography.
\newblock {\em Math. Methods Appl. Sci.}, 15(5):321--330, 1992.

\bibitem{nuyts1999}
J.~Nuyts, P.~Dupont, S.~Stroobants, R.~Benninck, L.~Mortelmans, and P.~Suetens.
\newblock Simultaneous maximum a posteriori reconstruction of attenuation and
  activity distributions from emission sinograms.
\newblock {\em IEEE Trans. Med. Imaging}, 18(5):393--403, 1999.

\bibitem{ollinger1997}
J.~M. Ollinger and J.~A. Fessler.
\newblock Positron-emission tomography.
\newblock {\em IEEE Signal Process. Mag.}, 14(1):43--55, 1997.

\bibitem{osman2003}
M.~M. Osman, C.~Cohade, Y.~Nakamoto, and R.~L. Wahl.
\newblock Respiratory motion artifacts on {PET} emission images obtained using
  {CT} attenuation correction on {PET-CT}.
\newblock {\em Euro. J. Nuc. Med. Mol. Imaging}, 30:603--606, 2003.

\bibitem{pierce2015}
L.~A. Pierce, B.~F. Elston, D.~A. Clunie, D.~Nelson, and P.~E. Kinahan.
\newblock {A digital reference object to analyze calculation accuracy of PET
  standardized uptake value}.
\newblock {\em Radiology}, 277(2):538--545, 2015.

\bibitem{ren2024}
Z.~Ren, E.~Y. Sidky, R.~F. Barber, C.~Kao, and X.~Pan.
\newblock Simultaneous activity and attenuation estimation in {TOF-PET} with
  {TV}-constrained nonconvex optimization.
\newblock {\em IEEE Trans. Med. Imaging}, 43(6):2347--2357, 2024.

\bibitem{rezaei2012}
A.~Rezaei, M.~Defrise, G.~Bal, C.~Michel, M.~Conti, C.~Watson, and J.~Nuyts.
\newblock Simultaneous reconstruction of activity and attenuation in
  time-of-flight {PET}.
\newblock {\em IEEE Trans. Med. Imaging}, 31(12):2224--2233, 2012.

\bibitem{rezaei2014}
A.~Rezaei, M.~Defrise, and J.~Nuyts.
\newblock {ML}-reconstruction for {TOF-PET} with simultaneous estimation of the
  attenuation factors.
\newblock {\em IEEE Trans. Med. Imaging}, 33(7):1563--1572, 2014.

\bibitem{segars2010}
W.~P. Segars, G.~Sturgeon, S.~Mendonca, J.~Grimes, and B.~M. Tsui.
\newblock {4D XCAT} phantom for multimodality imaging research.
\newblock {\em Med. Phys.}, 37(9):4902--4915, 2010.

\bibitem{Vaquero2015Positron}
J.~J. Vaquero and P.~Kinahan.
\newblock {Positron Emission Tomography: Current Challenges and Opportunities
  for Technological Advances in Clinical and Preclinical Imaging Systems}.
\newblock {\em Annu. Rev. Biomed. Eng.}, 17:385--414, 2015.

\bibitem{wang2006}
W.~Wang, Z.~Hu, E.~E. Gualtieri, M.~J. Parma, E.~S. Walsh, D.~Sebok, Y.~L.
  Hsieh, C.~H. Tung, X.~Song, J.~J. Griesmer, et~al.
\newblock Systematic and distributed time-of-flight list mode {PET}
  reconstruction.
\newblock In {\em 2006 IEEE Nucl. Sci. Symp. Conf. Record}, volume~3, pages
  1715--1722. IEEE, 2006.

\bibitem{zhou2004}
Z.~Wang, A.~C. Bovik, H.~R. Sheikh, and E.~P. Simoncelli.
\newblock Image quality assessment: from error visibility to structural
  similarity.
\newblock {\em IEEE Trans. Image Process.}, 13(4):600--612, 2004.

\bibitem{welch1997}
A.~Welch, R.~Clack, F.~Natterer, and G.~T. Gullberg.
\newblock Toward accurate attenuation correction in {SPECT} without
  transmission measurements.
\newblock {\em IEEE Trans. Med. Imaging}, 16(5):532--541, 1997.

\bibitem{xia2011}
T.~Xia, A.~M. Alessio, B.~De~Man, R.~Manjeshwar, E.~Asma, and P.~E. Kinahan.
\newblock Ultra-low dose {CT} attenuation correction for {PET/CT}.
\newblock {\em Phys. Med. Biol.}, 57(2):309--328, 2011.

\bibitem{xu2013}
Y.~Xu and W.~Yin.
\newblock A block coordinate descent method for regularized multiconvex
  optimization with applications to nonnegative tensor factorization and
  completion.
\newblock {\em SIAM J. Imaging Sci.}, 6(3):1758--1789, 2013.

\bibitem{zaidi2006}
H.~Zaidi and B.~H. Hasegawa.
\newblock Attenuation correction strategies in emission tomography.
\newblock In {\em Quantitative analysis in nuclear medicine imaging}, pages
  167--204. Springer, 2006.

\end{thebibliography}

\end{document}